\newif\ifsubmission
\submissiontrue
\submissionfalse %

\newif\ifcomment
\commenttrue
\commentfalse %

\newif\ifstoc 
\stoctrue
\stocfalse %

\ifstoc
\documentclass[sigconf,screen]{acmart}
\copyrightyear{2026}
\acmYear{2026}
\setcopyright{cc}
\setcctype{by}
\acmDOI{10.1145/3798129.3800836}
\acmISBN{979-8-4007-2536-4/2026/06}
\acmConference[STOC '26]{Proceedings of the 58th Annual ACM Symposium on Theory of Computing}{June 22--26, 2026}{Salt Lake City, UT, USA}
\acmBooktitle{Proceedings of the 58th Annual ACM Symposium on Theory of Computing (STOC '26), June 22--26, 2026, Salt Lake City, UT, USA}
\acmSubmissionID{stoc26main-p571-p}
\received{2025-11-02}
\received[accepted]{2026-02-01}

\else
\documentclass[11pt]{article}
\fi

\ifsubmission
\commentfalse
\author{Anonymous Authors}
\else
\ifstoc
\author{Yotam Kenneth-Mordoch}
\orcid{0000-0002-9212-9172}
\affiliation{%
  \institution{Weizmann Institute of Science}
  \city{Rehovot}
  \country{Israel}
}
\email{yotam.kenneth@weizmann.ac.il}

\author{Robert Krauthgamer}
\orcid{0009-0003-8154-3735}
\affiliation{%
  \institution{Weizmann Institute of Science}
  \city{Rehovot}
  \country{Israel}
}
\email{robert.krauthgamer@weizmann.ac.il}
\titlenote{The Harry Weinrebe Professorial Chair of Computer Science. Work partially supported by the Israel Science Foundation grant \#1336/23.}
\else
  \author{
    Yotam Kenneth-Mordoch
    \qquad
    Robert Krauthgamer%
    \thanks{ The Harry Weinrebe Professorial Chair of Computer Science.
      Work partially supported by the Israel Science Foundation grant \#1336/23.
    }
    \\ Weizmann Institute of Science
    \\ \texttt{\{yotam.kenneth,robert.krauthgamer\}@weizmann.ac.il}
  }
\fi
\fi

\ifstoc
\usepackage{microtype}
\else
\usepackage[hypertexnames=false]{hyperref}
\hypersetup{colorlinks=true,allcolors=blue}
\usepackage[letterpaper, margin=1in]{geometry}

\fi

\usepackage{amsthm}
\usepackage{algorithm}
\usepackage[noend]{algpseudocode}

\usepackage{makecell,hhline}
\usepackage{amsfonts}
\usepackage{import} 
\usepackage{xifthen}
\usepackage{mathtools}
\usepackage{arydshln}%
\usepackage{dsfont}
\usepackage[cal=esstix]{mathalfa}
\usepackage{thm-restate}
\usepackage{svg}
\usepackage{comment}
\usepackage[colorinlistoftodos,textsize=tiny,textwidth=2cm,color=red!25!white,obeyFinal]{todonotes}

\usepackage[capitalise,noabbrev]{cleveref} 

\ifdefined\AddToHook
\AddToHook{cmd/appendix/before}{\crefalias{section}{appendix}}
\fi

\newcommand{\mintcut}{\mathrm{cut}}

\newcommand{\R}{\mathbb{R}}

\newcommand{\vol}[1]{\text{vol}(#1)}
\newcommand{\G}{\mathcal{G}}
\newcommand{\f}{\varphi}

\newcommand{\tO}{\tilde{O}}

\DeclareMathOperator{\poly}{poly}
\DeclareMathOperator{\polylog}{polylog}
\providecommand{\set}[1]{{\{#1\}}}
\newcommand{\eqdef}{\coloneqq}

\newcommand{\colnote}[3]{\textcolor{#1}{$\ll$\textsf{#2}$\gg$\marginpar{\tiny\bf #3}}}

\ifcomment
\newcommand{\rnote}[1]{\colnote{purple}{#1--Robi}{RK}}
\newcommand{\ynote}[1]{\colnote{blue}{#1--Yotam}{YK}}
\else
\newcommand{\rnote}[1]{}
\newcommand{\ynote}[1]{}

\fi

\newtheorem{theorem}{Theorem}[section]

\newtheorem{claim}[theorem]{Claim}

\newtheorem{lemma}[theorem]{Lemma}
\newtheorem{corollary}[theorem]{Corollary}

\newtheorem{definition}[theorem]{Definition}

\newtheorem{fact}[theorem]{Fact}
\newtheorem{remark}[theorem]{Remark}

\newcommand{\cross}[2]{\crosscut_{#2}(#1)}

\DeclareMathOperator{\crosscut}{cr}
\newcommand{\sparsifiername}{\text{APMC sparsifier}}

\newcommand{\ap}{\textnormal{ap}}
\newcommand{\fr}{\textnormal{fr}}
\newcommand{\friendly}{H_{\fr}}
\newcommand{\allpairs}{H_{\ap}}
\newcommand{\cutspar}{H_{\textnormal{c}}}
\newcommand{\powercut}{H_{\textnormal{p}}}
\newcommand{\A}{\mathcal{A}}
\usepackage{array} %

\newcolumntype{P}[1]{>{\centering\arraybackslash}p{#1}}
\newcolumntype{M}[1]{>{\centering\arraybackslash}m{#1}}

\title{Faster All-Pairs Minimum Cut: Bypassing Exact Max-Flow}
\ifstoc
\titlenote{A full version of this paper is available at \href{https://arxiv.org/abs/2511.10036}{arXiv:2511.10036}.}
\fi
\begin{document}

\ifstoc
\else
\maketitle
\begin{abstract}
All-Pairs Minimum Cut (APMC) is a fundamental graph problem
that asks to find a minimum $s,t$-cut for every pair of vertices $s,t$.
A recent line of work on fast algorithms for APMC 
has culminated with a reduction of APMC to $\mathrm{polylog}(n)$-many max-flow computations.
But unfortunately, no fast algorithms are currently known for exact max-flow 
in several standard models of computation,
such as the cut-query model and the fully-dynamic model.

Our main technical contribution is a sparsifier
that preserves all minimum $s,t$-cuts in an unweighted graph,
and can be constructed using only approximate max-flow computations.
We then use this sparsifier to devise new algorithms for APMC in unweighted graphs
in several computational models:
(i) a randomized algorithm that makes $\tilde{O}(n^{3/2})$ cut queries to the input graph; 
(ii) a deterministic fully-dynamic algorithm with $n^{3/2+o(1)}$ worst-case update time; and
(iii) a randomized two-pass streaming algorithm with space requirement $\tilde{O}(n^{3/2})$.
These results improve over the known bounds,
even for (single pair) minimum $s,t$-cut in the respective models.
\end{abstract}

\fi

\ifstoc\else
\setcounter{page}{0}
\thispagestyle{empty}
\newpage
\fi
\section{Introduction}
\label{sec:introduction}
A powerful technique for speeding up the computation of graph algorithms is to employ sparsifiers, which are graphs that capture some property of the original graph while being significantly smaller.
For example, Nagamochi and Ibaraki \cite{NI92} showed that given a graph on $n$ vertices and a parameter $k>1$, one can construct a subgraph with $O(nk)$ edges that preserves all its cuts with at most $k$ edges.
This sparsifier is a key ingredient in several fast algorithms,
from finding a global minimum cut \cite{NI92},
through constructing $(1\pm\epsilon)$-cut sparsifiers \cite{FHHP19},
to solving the all-pairs minimum cut problem \cite{JPS21,AKT21a,AKT21b,AKT22b}.
Another example are non-trivial minimum cut (NMC) sparsifiers, introduced by Kawarabayashi and Thorup \cite{KT19}, which preserve all non-trivial minimum cuts of a graph while reducing the number of vertices in the graph to $\tilde{O}(n/\delta)$, where $\delta$ is the minimum degree of the graph
and $\tO(\cdot)$ hides polylogarithmic factors.
This sparsifier has been the basis of many recent algorithms
for finding a global minimum cut
in several models of computation \cite{KT19,HRW20,GNT20,AEGLMN22,GHNSTW23,KK25a,HKMR25,KK25c}.
Additional examples include friendly-cut sparsifiers \cite{AKT22b},
approximate cut sparsifiers \cite{BK96},
and many more variants that preserve other properties,
like distances (called spanners), flows, spectrum, and resistances;
or preserve such properties only for terminal vertices (called vertex sparsifiers).

We focus on the all-pairs minimum cut problem (APMC),
where the input is an edge-weighted graph $G=(V,E,w)$,
and the goal is to find for every pair $s,t\in V$
the minimum $s,t$-cut value, which is formally defined as 
\[
  \lambda_G(s,t)\eqdef \min \big\{ \mintcut_G(S) :\ {\{s\} \subseteq S \subseteq V\setminus\{t\}}  \big\} ,
\]
where $\mintcut_G(S)$ denotes the total weight of edges crossing the cut $(S,V\setminus S)$.
Throughout, we restrict attention to the case where the input $G$ is unweighted,
i.e., all edges have unit weight. 
Our method is based on efficiently constructing a sparsifier that preserves all these minimum cuts, and then solving it on the sparsifier.
A key difference from prior work is that our approach adds new vertices to the graph, 
as formalized in the next definition.

\begin{definition}
An all-pairs minimum cut sparsifier (APMC sparsifier) of a graph $G=(V,E)$ 
is a weighted graph $H=(U,E_H,w)$ where $U\supseteq V$,
such that for every $s,t\in V$,
there exists a minimum $s,t$-cut $S_H^{s,t}$ in $H$
such that $S_H^{s,t}\cap V$ is a minimum $s,t$-cut in $G$ of the same value, i.e.,  
\begin{equation*}
  \lambda_{H}(s,t)
  = \mintcut_H(S_H^{s,t})
  = \mintcut_G(S_H^{s,t}\cap V)
  = \lambda_{G}(s,t)
  .
\end{equation*}
\end{definition}

A well-known example of an APMC sparsifier 
is the celebrated Gomory-Hu tree~\cite{GH61},
which is a tree built on the same vertices ($U=V$ in our notation)
using $n-1$ calls to a minimum $s,t$-cut procedure.
Besides the obvious sparsity of a tree ($n-1$ edges), 
their algorithm achieves a quadratic improvement
over the $\binom{n}{2}$ calls needed by a straightforward computation for all pairs.
A recent line of work shows that a Gomory-Hu tree can be computed
using only $\polylog(n)$ calls to a minimum $s,t$-cut procedure \cite{AKT20,JPS21,AKT21a,AKT21b,AKT22b,AKLPST22,ALPS23,AKL+25,GKYY25}
and minimal additional computation, 
hence the complexity of solving all-pairs minimum cuts
is essentially equivalent to that of solving a single-pair minimum cut
(known to be equivalent to solving maximum $s,t$-flow, and thus called a max-flow computation).

Unfortunately, computing exact minimum $s,t$-cut is currently still expensive in several models of computation including the cut-query, streaming and fully-dynamic models.
Therefore, constructing a Gomory-Hu tree in these settings remains expensive as well.
Our main contribution is a new construction of an $\sparsifiername$
that uses only \emph{approximate} minimum $s,t$-cut,
yielding efficient algorithms in models where finding approximate cuts is cheap,
such as the cut-query and streaming models.
Furthermore, our APMC sparsifier is robust to changes in the input $G$
and hence can be maintained efficiently in the fully-dynamic model.
In fact, our sparsifier is so succinct that it yields algorithms for all-pairs minimum cut
that are more efficient than current algorithms for single-pair minimum cut in the aforementioned settings.

\subsection{Results}

\paragraph*{Cut-Query, Two-Party Communication and Submodular Minimization}
Our first result is an improved algorithm for all-pairs minimum cut in the cut-query model.
It not only improves over the $\tO(n^{7/4})$ query complexity known for this problem \cite{KK25b},
but also the $\tO(n^{8/5})$ query complexity known for (single pair) minimum $s,t$-cut \cite{JNS25}.
\begin{theorem}
    \label{theorem:main-result-cut-query}
    There exists a randomized algorithm, that given cut-query access to an unweighted graph $G$ on $n$ vertices, solves the all-pairs minimum cut problem using $\tO(n^{3/2})$ cut queries and succeeds with probability $1-1/\poly(n)$.
\end{theorem}
By an easy reduction, this result also implies
the first (non-trivial) algorithm for all-pairs minimum cut in the two-party communication model.
In this model, the edges of the input graph are partitioned between two players that wish to jointly compute some function of the graph.
The goal is to design a communication protocol that allows the computation while minimizing the total number of bits exchanged between the two players.
We note that our protocol matches the known bound for (single-pair) minimum $s,t$-cut \cite{JNS25,GJ25}. 
\begin{corollary}
    \label{corollary:two-party-comm-result}
    There exists a randomized two-player communication protocol that computes all-pairs minimum cut problem of an unweighted graph with $n$ vertices using $\tO(n^{3/2})$ bits of communication and succeeds with probability $1-1/\poly(n)$.
\end{corollary}

Furthermore, recent work implies
that one can compute a Gomory-Hu tree of a symmetric submodular function
using $\polylog (n)$ calls to a submodular minimization procedure
on functions with total ground-set size $\tO(n)$ \cite{GKYY25}.%
\footnote{The original work focuses on reducing the construction of Gomory-Hu trees for graphs to $\polylog n$ max-flow instances, we verify that their approach extends to all symmetric submodular functions.}
It then follows easily that the query complexity APMC is at most $\polylog(n)$ times
the query complexity of minimizing a submodular function (abbreviated SFM) on $n$ elements.%
\footnote{The algorithms solves many small SFM instances on contraction of the ground set, whose total size is $\tO(n)$.}
\ifstoc
\else
For completeness, we state this formally in the next theorem and prove it in \Cref{sec:gomory-hu-submodular}.
\fi
It basically rules out the natural approach of proving a query lower bound for SFM
that is based on the hardness of APMC in the cut-query model
and significantly exceeds $\Omega(n^{3/2})$,
as that would contradict our cut-query algorithm above.

\begin{restatable}[Query Complexity Version of \cite{GKYY25}]{theorem}{ghsubmodular}
    \label{theorem:submodular-gomory-hu}
    Let $Q(n)$ be the query complexity of minimizing a symmetric submodular function over a ground set of size $n$, and assume $Q(n)=n\cdot g(n)$ for some non-decreasing function $g$.
    Then, there exists an algorithm for constructing a Gomory-Hu tree of a symmetric submodular $f:2^{[n]}\to \R_+$ using $Q(n)\cdot \polylog (n)$ queries.
\end{restatable}

\paragraph*{Fully-Dynamic Algorithms}
In the fully-dynamic setting, the input is a \emph{dynamic graph} presented as
a sequence of edge insertions and deletions over a fixed set of $n$ vertices,
and the goal is to maintain some function of the graph after each edge update. 
We present the first (non-trivial) fully-dynamic algorithm for APMC
(i.e., for maintaining a Gomory-Hu tree).
Furthermore, our algorithm is deterministic, and thus robust to adversarial updates.
Even for (single-pair) minimum $s,t$-cut, all previous algorithms
either return an approximate solution \cite{ADKKP16,GRST21,BBG+22}
or are limited to instances where $\lambda_G(s,t)\le\log^{o(1)}n$ \cite{GI91,Frederickson97,HK99,JS21}.

\begin{theorem}
\label{theorem:main-result-fully-dynamic}
There exists a deterministic fully-dynamic algorithm
that maintains a Gomory-Hu tree for a dynamic unweighted graph on $n$ vertices, 
using $n^{3/2+o(1)}$ worst-case update time.
\end{theorem}

This theorem immediately implies an algorithm with $n^{3/2+o(1)}$ worst-case update time
for maintaining edge connectivity in a dynamic graph.
This is the first deterministic algorithm for the problem with $o(n^2)$ worst-case update time.
Furthermore, when $m\ge n^{37/22}\approx n^{1.7}$,
it matches, up to subpolynomial factors,
the state-of-the-art deterministic amortized-time complexity
\ifstoc
\else
of $\tO(\min(m^{1-1/12}, m^{11/13}n^{1/13}, n^{3/2}))$
\fi
\cite{VC25}.

\begin{corollary}
There exists a deterministic fully-dynamic algorithm that
maintains the value of a global minimum cut of a dynamic unweighted graph on $n$ vertices,
using $n^{3/2+o(1)}$ worst-case update time.
\end{corollary}

\paragraph*{Streaming Algorithms}
We provide the first (non-trivial) algorithm for all-pairs minimum cut
in dynamic streams,
where the graph is presented as a sequence of edge insertions and deletions.
Our algorithm uses $\tO(n^{3/2})$ space and makes two pass. 
Previously, the only known algorithm was for finding a (single-pair) minimum $s,t$-cut,
and it uses $\tO(n^{5/3})$ space in two passes \cite{RSW18}.
No bounds are known even with other restrictions,
e.g., insertion-only (no deletions) or another number of passes.

\begin{theorem}
    \label{theorem:main-result-streaming}
    Given an unweighted graph $G$ on $n$ vertices presented as a dynamic stream, 
    one can solve the all-pairs minimum cut problem using $\tO(n^{3/2})$ bits of storage in two passes.
    The algorithm is randomized and succeeds with probability $1-1/\poly(n)$.
\end{theorem}

A summary of our results is presented in \Cref{tab:complexity}.
\begin{table}[ht]
    \centering
    \caption{Comparison of our algorithms for all-pairs minimum-cut with existing algorithms for the same problem and for single pair minimum $s,t$-cut, across different computational models.}
    \label{tab:complexity}
    \renewcommand{\arraystretch}{1.3}
    \ifstoc
    \footnotesize
    \setlength{\tabcolsep}{2pt}
    \begin{tabular}{|M{1.5cm}|M{1.2cm}|M{1.6cm}|M{2.8cm}|}
    \else
    \begin{tabular}{|M{2.9cm}|M{2.4cm}|M{3.4cm}|M{6cm}|}
    \fi
        \hline
        \textbf{Model} &
        \textbf{Measure} &
        \textbf{Our Result} &
        \textbf{Known Results} \\
        \hline
        cut query & queries & $\tilde{O}(n^{3/2})$ Thm.~\ref{theorem:main-result-cut-query}
        & \makecell{$\tilde{O}(n^{8/5})$, single pair \cite{JNS25}
        \\ $\tilde{O}(n^{7/4})$, all pairs \cite{KK25b}}
        \\
        \hline
        fully dynamic & worst-case update time & $n^{3/2+o(1)}$ Thm.~\ref{theorem:main-result-fully-dynamic}
        & $m^{1+o(1)}$, single pair \\
        \hline
        dynamic streaming & storage, two~passes & $\tilde{O}(n^{3/2})$ Thm.~\ref{theorem:main-result-streaming}
        & \makecell{$\tilde{O}(n^{5/3})$, single pair \cite{RSW18} \\ $\tilde{O}(m)$, all pairs} \\
        \hline
        two-party comm. & bits & $\tilde{O}(n^{3/2})$ Cor.~\ref{corollary:two-party-comm-result}
        & \makecell{$\tilde{O}(n^{3/2})$, single pair \cite{GJ25}\footnotemark \\ $\tilde{O}(m)$, all pairs} \\
        \hline
    \end{tabular}
\end{table}
\footnotetext{The results of \cite{GJ25} also apply to weighted graphs.}

\subsection{Related Work}

\paragraph*{All-Pairs Minimum Cut}
The all-pairs minimum cut problem has received significant attention in recent years, culminating in algorithms that make $O(\polylog n)$ calls to a minimum $s,t$-cut procedure and require in addition $\tilde{O}(m)$ time \cite{HKP07,BHKP07,AKT20,JPS21,AKT21a,AKT21b,AKT22b,AKLPST22,ALPS23,AKL+25,GKYY25}.
Over roughly the same time span, there has been significant progress on faster algorithms (in running time) for the minimum $s,t$-cut problem itself,
and currently the best algorithm runs in $m^{1+o(1)}$ time \cite{CKLPGS25}.  

\paragraph{Cut-Query Algorithms}
The study of cut-query algorithms was initiated in \cite{RSW18}, motivated by the relation between cut queries and submodular minimization.
For the global minimum cut problem, known algorithms in unweighted graphs use $O(n)$ randomized cut queries \cite{RSW18,AEGLMN22}, matching the lower bound \cite{GPRW20, LLSZ21}, and $\tO(n)$ randomized cut queries suffice in weighted graphs \cite{MN20};
in addition, known deterministic algorithms use $\tO(n^{5/3})$ queries in unweighted graphs \cite{ASW25}.
For minimum $s,t$-cut in unweighted graphs, known algorithms use $\tO(n^{8/5})$ cut queries \cite{RSW18,ASW25,JNS25},
leaving a gap as the only lower bound is $\Omega(n)$ queries.
Recently, it was shown that APMC in unweighted graphs can be solved using $\tO(n^{7/4})$ cut queries \cite{KK25b}.
Other related work include an algorithm for approximating the maximum cut in weighted graphs
using $\tO(n)$ cut queries \cite{PRW24},
and studying parallel efficiency by measuring the number of query rounds 
for both global and $s,t$ minimum cuts in \cite{KK25a}.

\paragraph{Cuts in Dynamic Streams}
In the streaming model, exact computation of a global (or $s,t$) minimum cut requires $\Omega(n^2)$ space in one pass \cite{Zelke11}.%
\footnote{These lower bounds do not apply in random order streams where better results are possible, e.g. \cite{DGL+25}.}
Consequently, research has primarily focused on approximate solutions;
one can obtain $(1\pm\epsilon)$-approximation using $\tilde{O}(\epsilon^{-2}n)$ space by employing cut sparsifiers, see e.g.\ \cite{AGM12}.
For global minimum cut in unweighted graphs this was later improved to $\tO(\epsilon^{-1}n)$ \cite{DGL+25}.
The study of exact solutions using more passes was motivated by an observation in \cite{ACK19} that the approach of \cite{RSW18}
for global and $s,t$ minimum cut in unweighted graphs
can be implemented in two passes
using $\tO(n)$ space and $\tO(n^{5/3})$ space, respectively;
see also \cite{AD21}.
For weighted graphs,
an $O(\log n)$-pass algorithm for global minimum cut using $\tO(n)$ space was presented in \cite{MN20},
and a smooth tradeoff of $\tO(rn^{1+1/r})$ space in $O(r)$ passes was later shown in \cite{KK25a}.
Finally, for finding a minimum $s,t$-cut in weighted graphs,
there exists a lower bound of $\Omega(n^2/p^5)$ space for $p$-pass streaming algorithms \cite{ACK19}.

\paragraph{Cuts in Fully-Dynamic Graphs}
The study of global minimum cut in the fully-dynamic setting has a rich history: it was originally studied for small values of edge connectivity,
namely $\lambda_G \eqdef \min_{s,t} \lambda_G(s,t) \le \log^{o(1)}n$,
see e.g.\ \cite{EGIN97,HLT01,KKM13,CGLNPS20,JST24}.
For general $\lambda_G$, there exists a deterministic algorithm that
maintains $\lambda_G$
in worst-case update time $\tO(\min\left(\lambda_{\max}^{5.5}\sqrt{n}\right))$,
where $\lambda_{\max}$ is the maximum edge connectivity of $G$ throughout the graph updates \cite{Thorup07,GHNSTW23,VC25}.
Allowing randomization, there exist algorithms that maintain $\lambda_G$ in worst-case update time $\tO(\min\set{n,(n/\lambda_G^2)^2})$ \cite{GHNSTW23,HKMR25,KK25c}.
Finally, one can achieve a $(1+o(1))$-approximation with worst-case update time $n^{o(1)}$ \cite{TK00,Thorup07,EHL25}.
The study of minimum $s,t$-cut in the fully dynamic setting has received less attention,
but features similar tradeoffs as global minimum cut: 
Existing algorithms either return an approximate solution \cite{ADKKP16,GRST21,BBG+22}
or are limited to instances where $\lambda_G(s,t)\le\log^{o(1)}n$ \cite{GI91,Frederickson97,HK99,JS21}.

\subsection{Future Work}

\paragraph*{Lower Bounds for Submodular Function Minimization}
One of the main motivations for studying APMC in the cut-query model
is its connection to submodular function minimization (SFM).
Graph-cut functions are natural examples of submodular functions and have been extensively studied in the context of SFM.
In particular, the best lower bounds known for SFM are based on graph-cut functions \cite{GPRW20, LLSZ21} (there are slightly stronger lower bounds for deterministic algorithms that are based on different functions \cite{CGJS23}).
The following lemma shows that maximum-flow algorithms in the cut-query model
have an inherent lower bound of $\Omega(n^{3/2})$ queries. 
It is based on a graph constructed in \cite{KL98},
where the maximum $s,t$-flow contains $\Omega(n^{3/2})$ edges,
and has been observed previously \cite{ASW25,JNS25} without a formal proof,
\ifstoc
\else
hence we prove it for completeness in \Cref{sec:max-flow-cut-query-lower-bound}.
\fi

\begin{lemma}
\label{lemma:maximum-s-t-flow-lower-bound}
Every randomized cut-query algorithm that recovers all edges of a maximum $s,t$-flow in an unweighted graph $G$ and succeeds with high  constant probability requires $\Omega(n^{3/2}/\log n)$ cut queries.
\end{lemma}

Previous algorithms for minimum $s,t$-cut in the cut-query model \cite{RSW18,ASW25,JNS25}
are affected by this lower bound, e.g., they compute a maximum $s,t$-flow.
In contrast, our algorithm for all-pairs minimum cut circumvents this lower bound
by avoiding explicit recovery of any maximum $s,t$-flow.
However, it still does not manage to break the $O(n^{3/2})$ barrier,
which remains open.
We conjecture that this is not possible
and that every algorithm for minimum $s,t$-cut in the cut-query model requires $\tilde{\Omega}(n^{3/2})$ queries.
Finally, \Cref{theorem:submodular-gomory-hu}
(which constructs a Gomory-Hu tree for a symmetric submodular function
using $\polylog(n)$ calls to an SFM procedure on ground sets of size $n$)
yields an exciting possible avenue for proving lower bounds for SFM.

\paragraph*{Extension to Weighted Graphs}
The structural results that form the basis of our APMC sparsifier
do \emph{not} seem to extend to weighted graphs.
In the case of streaming algorithms,
there is an $\Omega(n^2/p^5)$ space lower bound for $p$-pass algorithms \cite{ACK19},
which rules out most practical algorithms.
However, in the cut-query and fully-dynamic models,
no such lower bounds are known,
and in fact we are still lacking non-trivial algorithms
for (single pair) minimum $s,t$-cut in weighted graphs;
thus, either algorithms or lower bounds would be very interesting.

\section{Technical Overview}
\label{sec:technical-overview}
Our work contains two main structural insights.
The first one is that the minimum $s,t$-cuts of a graph $G$ (for all pairs $s,t$) are encoded by the so-called friendly cuts of $G$ and the vertex degrees in $G$.
The concept of friendly cuts  
emerges naturally in social learning, statistical physics and set theory
under various names \cite{Morris00,GK00,BTV06,BL16,GMT20,FKN+22}.
Informally, a cut $C\subseteq V$ is called friendly
if every vertex $v\in V$ has at least as many edges to its side of the cut
as to the other side (see \cref{sec:structural-proof} for the formal definition). 
Throughout, let $E(A,B)$ denote the set of edges in $G$
with one endpoint in $A\subseteq V$ and the other in $B\subseteq V$;
when $A=\{v\}$, we simply write $E(v,B)$ instead of $E(\{v\},B)$.
In this notation, $C$ is friendly
if $|E(v,C)| \ge |E(v,V\setminus C)|$ for all $v\in C$,
and the reverse inequality for all $v\in V\setminus C$.

Recently, it was shown that one can take any unweighted graph $G$
and contract some of its vertices to obtain a graph $\friendly$,
such that the value of every friendly cut in $G$ is preserved,
i.e., for every friendly cut $C\subseteq V$
we have $\mintcut_G(C)=\mintcut_{\friendly}(C)$,
where $C$ is interpreted also as a cut in $\friendly$ in the natural way.
We call such a graph $\friendly$ a \emph{friendly cut sparsifier}
(see \Cref{definition:friendly-cut-sparsifier}).
A \emph{contraction} of $G=(V,E)$ is a graph that is obtained from $G$
by partitioning $V$ into subsets $u_1,\ldots,u_k\subseteq V$
and contracting each subset $u_i$ to form a \emph{super-vertex},
keeping only the edges that connect different subsets,
which may now become parallel edges.
To illustrate this, consider a dumbbell graph,
which consists of two cliques connected by an edge;
contracting each clique into a super-vertex yields
a graph consisting of a single edge,
which is clearly a friendly cut sparsifier.

Throughout the paper we assume that the graph contains no vertices of degree $1$ to simplify the proofs.
This limitation can be easily removed in all our algorithmic results by handling such vertices separately.%
\footnote{We can also incorporate vertices of degree $1$ into our cut encoding and APMC sparsifier
  (and thus handle them directly and not separately in our algorithmic results),
  however this requires additional technical details which we omit for clarity.
}
Our first structural result is the following.

\begin{theorem}
\label{theorem:encoding-structural-result}
Let $G=(V,E)$ be an unweighted graph
and let $\A_{s,t}\subseteq 2^V$ denote the set of minimum $s,t$-cuts in $G$ for $s,t\in V$.
Given a friendly cut sparsifier of $G$ and the vertex degrees of $G$,
one can recover for each pair $s,t\in V$ at least one cut $S\in \A_{s,t}$ along with its value.
\end{theorem}

\begin{remark}
\Cref{theorem:encoding-structural-result} only guarantees recovering at least one minimum $s,t$-cut for each pair $s,t$.
However, if no vertex of degree at most $2$ in $G$ is contracted during the construction of the friendly cut sparsifier,
then one can actually recover the entire set $\A_{s,t}$ for each pair $s,t$.
We omit the details for brevity.
\end{remark}

Friendly cut sparsifiers were previously used to design
fast algorithms for APMC in \cite{AKT22b,KK25b}.
Their approach is to divide the optimal cuts into friendly and unfriendly;
friendly cuts are found efficiently by applying known algorithms on the sparsifier,
whereas unfriendly cuts admit fast algorithms that find them directly in the original graph $G$.
In contrast, we show that the friendly cut sparsifier already encodes all the minimum $s,t$-cuts without further access to $G$ except for the vertex degrees.
Unfortunately, recovering the minimum $s,t$-cuts from this encoding naively requires iterating over all the cuts of $\friendly$, which might take exponential time.
While this may be acceptable in the streaming and cut-query settings,
where one is mostly concerned with storage and query complexities,
this approach is too inefficient for the fully-dynamic setting,
and here comes our second structural contribution,
which is to construct a small APMC sparsifier based on the encoding.
The advantage of the APMC sparsifier is that it is a graph,
and thus one can apply to it known Gomory-Hu tree construction algorithms,
and thereby recover all minimum $s,t$-cuts in the original graph efficiently.

Our APMC sparsifier construction is based on a new \emph{star transform} operation,
that rearranges the vertices forming a super-vertex 
of the friendly cut sparsifier 
into a star graph, as follows.
Given a super-vertex $u$ of $\friendly$,
split $u$ back into its constituent vertices $\set{v}_{v\in u}$,
and reconnect every edge incident to $u$ back to its original endpoint $v\in u$.
In addition, add a new \emph{proxy vertex} $u'$ and connect it to every vertex $v\in u$ with edge weight $w(v,u') \eqdef |E(v,u\setminus v)|$,
which is the degree of $v$ in the induced subgraph $G[u]$. %
See \Cref{fig:splitting-process} for illustration.
We prove that performing this star transform to \emph{all super-vertices} in $\friendly$
yields a graph $\allpairs$ that is an APMC sparsifier.
We stress that all the information needed for the star transform
is encoded in $\friendly$ and the vertex degrees of $G$,
simply because $|E(v,u\setminus v)| = \deg(v)-|E(v,V\setminus u)|$
and the subtracted quantity counts edges that appear in $\friendly$,
and can be recovered since we store the edges of $\friendly$ along with their original endpoints in $G$.

Our star transform bears superficial resemblance to the star-mesh transform
from electrical network theory \cite{Bedrosian61,VO73},
which replaces a star with a weighted clique on the non-center vertices (of the star),
with edges weights that preserves the resistance between every pair of non-center vertices.
The special case of 3 non-center vertices is known as the Wye-Delta transform,
and has been used in cut sparsification of planar graphs \cite{KR14,GHP20,KK21}.
However, there are some crucial differences. 
First, our star transform increases the number of vertices in the graph
(recall it adds a proxy vertex), 
and thus it is more similar to the \emph{inverse} of the star-mesh transform,
which does not exist in general.
Second, the star-mesh transform preserves cuts only in planar graphs,
while our star transform is applicable to general graphs.

\begin{theorem}
\label{theorem:apmc-sparsifier-construction}
Let $G$ be an unweighted graph on $n$ vertices.
Given access to a friendly cut sparsifier $\friendly=(V_{\friendly}, E_{\friendly})$ of $G$
and to the vertex degrees of $G$,
one can construct an APMC sparsifier $\allpairs=(U,E_{\ap})$ of $G$
with $|E_{\ap}| \le |E_{\friendly}|+n$ in time $O(|E_{\friendly}|+n)$.
\end{theorem}

\begin{figure}[h]
    \centering
    \includegraphics[width=0.8\columnwidth]{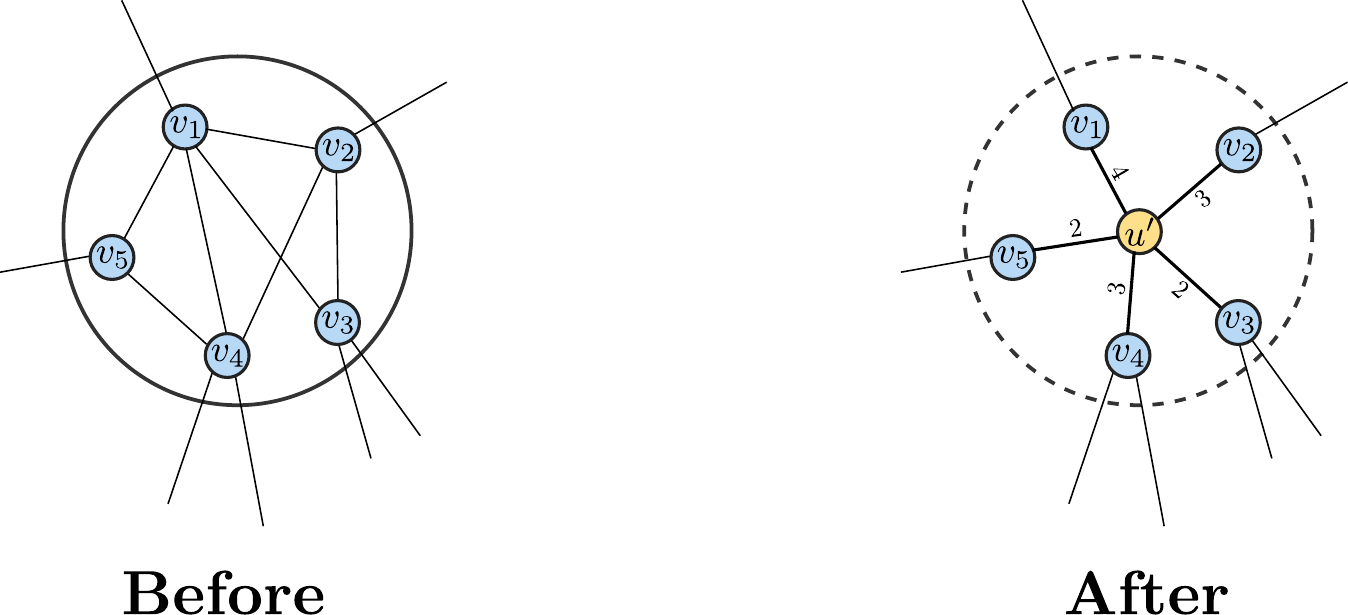}
    \caption{
        An illustration of the star-transform operation. 
        The super-vertex $u=\set{v_1,\ldots,v_5}$ is split into its constituent vertices, adding a proxy vertex $u'$ that is connected to every vertex $v_i\in u$ with edge weight equal to its degree in the induced graph $G[u]$.
        For example, $v_5$ has two neighbors in $G[u]$,
        hence it is connected to $u'$ with an edge of weight $2$.
        }
    \label{fig:splitting-process}
\end{figure}

One main advantage of our APMC sparsifier, in comparison to the Gomory-Hu tree, is that its construction does not require computing any \emph{exact} minimum $s,t$-cut in $G$.
Instead, it relies only on friendly cut sparsifiers,
which are based on expander decomposition,
and thus can be computed using \emph{approximate} minimum $s,t$-cuts.
This method suffices for the cut-query and streaming models,
where we can leverage existing algorithms for expander decomposition \cite{FKM23,CKMM24,KK25b}.
For the fully-dynamic model, we provide a new guarantee,
namely, that existing constructions of friendly cut sparsifiers are robust to edge updates, 
i.e., given a sequence of $\tO(\sqrt{n})$ edge updates to $G$ one can apply the same updates to $\friendly$
and obtain a friendly cut sparsifier of the updated graph.
This guarantee allows us to build a new friendly cut sparsifier from scratch every $\tO(\sqrt{n})$ updates,
and thus maintain an APMC sparsifier using $\tO(n^{3/2})$ worst-case update time.

The rest of the paper is organized as follows.
In \Cref{sec:structural-proof} we prove \Cref{theorem:encoding-structural-result,theorem:apmc-sparsifier-construction}.
We then use these structural results  in \Cref{sec:applications}
to design algorithms for the cut-query, fully-dynamic, and streaming models, proving \Cref{theorem:main-result-cut-query,theorem:main-result-fully-dynamic,theorem:main-result-streaming}.
\ifstoc
\else
Finally, we provide in \Cref{sec:clusters,sec:dynamic-streams}
more details for the construction of friendly cut sparsifiers in the streaming and fully-dynamic models.
\fi

\section{Encoding All Minimum $s,t$-Cuts}
\label{sec:structural-proof}
In these section we prove \Cref{theorem:encoding-structural-result,theorem:apmc-sparsifier-construction}, showing how to encode all the minimum $s,t$-cuts of a graph using a friendly cut sparsifier and the degree of every vertex, and how to construct an all-pairs minimum-cut (APMC) sparsifier using this encoding.

We begin by formally defining the notion of friendly cuts and friendly cut sparsifiers for a graph $G=(V,E)$. 
Given a cut $S\subseteq V$, let $\cross{v}{S}$ be
the number of edges incident to $v$ that cross the cut $E(S,V\setminus S)$, 
and let us omit the subscript when clear from context.
Define the \emph{friendliness ratio} of a vertex $v\in V$
(with respect to $S$) as $1-\cross{v}{}/\deg(v)$,
where $\deg(v)$ is the degree of $v$ in $G$.
The cut $S$ is called \emph{$\alpha$-friendly}
if every vertex $v\in V$ has friendliness ratio $1-\cross{v}{S}/\deg(v)\ge \alpha$,
and otherwise it is called \emph{$\alpha$-unfriendly}. 

Throughout, a \emph{contraction} of $G=(V,E)$ is a graph $H$
obtained from $G$ by contracting some subsets of vertices.
A \emph{super-vertex} is a vertex in $H$ formed by contracting multiple vertices in $G$. 
Observe that contracting two vertices $u,v\in V$
is equivalent to adding an edge of infinite capacity between them. 
This view is useful because now $H$ has the same vertex set as $G$,
and furthermore, cuts in $H$ are at least as large as those in $G$
(when comparing the same $S\subset V$).
All friendly cut sparsifiers throughout the paper are contraction-based.

\begin{definition}%
\label{definition:friendly-cut-sparsifier}
An \emph{$(\alpha,w)$-friendly cut sparsifier} of a graph $G$ is a contraction $H$ of $G$
such that for  every $\alpha$-friendly cut $S\subseteq V$ of $G$ with at most $w$ edges we have $\mintcut_H(S)=\mintcut_G(S)$.
\end{definition}
\begin{remark}
The statements of \Cref{theorem:encoding-structural-result,theorem:apmc-sparsifier-construction} do not specify the parameters of the friendly cut sparsifier,
both theorems use a $(1/6,2n)$-friendly cut sparsifier.
\end{remark}
The main idea in proving \Cref{theorem:encoding-structural-result}
(encoding all the minimum $s,t$-cuts using a friendly cut sparsifier)
is provided in \Cref{lemma:structure-min-st-cut} below,
which in turn relies on the following easy claim.
We say that a minimum $s,t$-cut $S\subseteq V$ with $s\in S$ is \emph{minimal}
if for all $S'\subset S$ with $s\in S$ we have $\mintcut_G(S')>\mintcut_G(S)$.
\begin{claim}
\label{claim:at-most-one-unfriendly}
Let $S\subseteq V$ be a minimum $s,t$-cut in $G$
that is $\alpha$-unfriendly for some $\alpha\le 1/6$.
Then, exactly one vertex in $G$ has friendliness ratio less than $1/6$,
and it is either $s$ or $t$.
\end{claim}

\begin{proof}
Every vertex besides $s,t$ must have friendliness ratio at least $1/2$,
as otherwise moving it to the other side of the cut would yield a cut with a smaller value,
contradicting the optimality of $S$.
The cut $S$ is $\alpha$-unfriendly for $\alpha\le 1/6$,
and thus at least one of $s,t$ has friendliness ratio less than $1/6$.
Assume towards contradiction this holds for both vertices,
then together with the optimality of $S$ we have
\[
  \cross{t}{S}+\cross{s}{S}
  > 5/6 \cdot [\deg(s) + \deg(t)]
  \geq 5/3 \cdot \mintcut(S).
\]
However, at most one edge in $E(S,V\setminus S)$ can be incident to both $s$ and $t$,
and hence $\cross{t}{S}+\cross{s}{S} \le \mintcut(S)+1$.
This yields a contradiction if $\mintcut(S)\ge 2$, 
which indeed holds because we assumed $G$ has no vertices of degree $1$ 
and thus $\mintcut(S)\ge \cross{s}{S} > 5/6\cdot \deg(s) \geq 5/3$. 
\end{proof}

\begin{lemma}
\label{lemma:structure-min-st-cut}
Let $S\subseteq V$ be a minimal minimum $s,t$-cut in $G$, 
and suppose $s\in S$ has friendliness ratio less than $1/6$.
Then $X=S\setminus \set{s}$ is a $\frac{1}{6}$-friendly cut in $G$
and furthermore $\mintcut_G(X) \le 2\deg(s)$.
\end{lemma}

\begin{proof}
First observe that if $S=\set{s}$ then $X=\emptyset$ is trivially a friendly cut.
Next, we show that every vertex $x\in X$ has $\deg(x)\ge 3$.
Indeed, we assumed $G$ has no vertices of degree $1$, 
and $x$ also cannot have degree $2$,
as otherwise moving it to the other side of $S$ would violate minimality of $S$;
hence, $\deg(x)\ge 3$.

We proceed to prove that every vertex in $G$ has friendliness ratio at least $1/6$
with respect to the cut $X$.
Consider first vertices $x\in X$. 
We know that $\cross{x}{S}/\deg(x)\le 1/2$,
as otherwise moving $x$ to the other side of the cut yields a cut with a smaller value, in contradiction to optimality of $S$.
Furthermore, $|\cross{x}{X} - \cross{x}{S}| \le 1$ as removing $s$ from $S$
can change the cut by at most one edge incident to $x$. 
Therefore, the friendliness ratio of $x$ with respect to $X$ is at least
$1-(\cross{x}{S}+1)/\deg(x)\ge 1-1/2 - 1/\deg(x)\ge 1/6$.

Next, consider vertices in $V\setminus X$.
Vertices other than $s,t$ have, by same argument as before,
friendliness ratio at least $1/2\ge 1/6$ in the cut $S$,
which can only increase when moving $s$ to the other side of the cut.
The vertex $t$ has, by \Cref{claim:at-most-one-unfriendly}, 
friendliness ratio at least $1/6$ in the cut $S$,
which can only increase when we move $s$ to the other side of the cut.
For the vertex $s$, its friendliness ratio in the cut $S$ is less than $1/6$,
and thus in the cut $X$ it is 
$1-\cross{s}{X}/\deg(s)= 1-(\deg(s)-\cross{s}{S})/\deg(s) > 5/6$.
This concludes the bound $1/6$ for all vertices in $G$. 
    
Finally, the optimality of the cut $S$ implies
$\mintcut_G(X) \leq \mintcut_G(S) + \deg(s) \le 2\deg(s)$.
\end{proof}

\begin{proof}[Proof of \Cref{theorem:encoding-structural-result}]
Let $\friendly=(V,E_{\fr},w_{\fr})$ be a $(1/6,2n)$-friendly cut sparsifier of $G$,
and let $w_{\fr}(\cdot,\cdot)$ denote the total weight of corresponding edges in $E_{\fr}(\cdot,\cdot)$.
To simplify notation, let 
\begin{equation}
\label{eq:fv}
  f_v(X)\eqdef \mintcut_{\friendly}(X)-2w_{\fr}(X,v)+\deg(v) ,
\end{equation}
for $v\in V$ and $X\subseteq V\setminus\{v\}$,
and notice that it can be computed given $\friendly$ and the vertex degrees in $G$. 
The algorithm to recover a minimum $s,t$-cut in $G$ is as follows.
Find a minimum $s,t$-cut in $\friendly$ using any standard algorithm,
and denote it by $S_{\fr}\subset V$.
In addition, let
\begin{equation*}
    X^{(s)}
    = \arg\min_{X\subseteq V\setminus\set{s,t}}f_s(X)
    ,
\end{equation*}
and define $X^{(t)}$ similarly.
Finally, return a set associated with the minimum of those three values, 
namely, one of $S_{\fr}, X^{(s)}\cup\set{s}, X^{(t)}\cup\set{t}$
together with its associated value.

We proceed to prove the correctness of the algorithm.
We have $\mintcut_G(A\cup B) = \mintcut_G(A) + \mintcut_G(B) - 2\cdot |E(A,B)|$
whenever $A,B\subseteq V$ are disjoint,
and thus the value of a cut $S=X\cup\set{v}$ where $X\subseteq V\setminus \set{v}$ can be written as 
\begin{equation}
  \label{eq:min-cut-decomposition}
  \mintcut_G(S)
    = \mintcut_G(X)
    -2\cdot|E(X,v)|
    +\deg(v)
  .
\end{equation}
Let us now show that $\mintcut_G(X\cup\set{v}) \le f_{v}(X)$
by comparing \eqref{eq:min-cut-decomposition} with \eqref{eq:fv} term by term. 
Observe that $\mintcut_G(X) \le \mintcut_{\friendly}(X)$ for all $X\subseteq V$,
as $\friendly$ is a contraction of $G$.
In addition, $w_{\fr}(X,v)$ has the same value as $|E(X,v)|$
up to infinite-weight edges in $\friendly$ (arising from contractions).
However, changing the weight of such edges does not affect $f_v(X)$,
because if $v$ is contracted with a vertex of $X$
then the weight of the edge between them
cancels out completely in $f_v(X)$ (it appears in all three terms).
We conclude that $\mintcut_G(X\cup\set{v}) \le f_{v}(X)$;
furthermore, equality holds if and only if
no edge of $E(X,V\setminus X)$ is contracted in $\friendly$.
In addition, the value of the cut $S_{\fr}$ in $\friendly$
(which is a contraction of $G$) 
is clearly at least as large as a minimum $s,t$-cut value in $G$.
Hence, the algorithm's output value is always at least as large as $\lambda_{s,t}(G)$.

We now argue that the algorithm output is actually $\lambda_{s,t}(G)$
(together with a minimum $s,t$-cut in $G$).
If there exists a minimum $s,t$-cut that is $1/6$-friendly,
then it is preserved in $\friendly$,
and thus the algorithm will find its value when computing $S_{\fr}$.
Otherwise, fix some $1/6$-unfriendly minimal minimum $s,t$-cut $S\subset V$ in $G$, 
and assume without loss of generality that $s\in S$ has friendliness ratio less than $1/6$.
By \Cref{lemma:structure-min-st-cut},
$X=S\setminus\set{s}$ is a $1/6$-friendly cut with at most $2n$ edges.
Therefore, the cut of $X$ is preserved in $\friendly$,
i.e., no edge in $E(X,V\setminus X)$ is contracted in $\friendly$
and hence $f_s(X)=\mintcut_G(X\cup\set{s}) = \lambda_{s,t}(G)$,
and the algorithm will find this value when it evaluates $X^{(s)}$.
\end{proof}

Unfortunately, recovering a minimum $s,t$-cut using \Cref{theorem:encoding-structural-result} requires iterating over all the cuts of $\friendly$,
potentially taking exponential time
as $\friendly$ might have $\Omega(n)$ vertices.
While this is acceptable in the cut-query and streaming models, where one is mostly concerned about the number of queries or space,
it is excessive for the fully-dynamic setting where fast update times is crucial.
We therefore propose to modify the friendly-cut sparsifier into
an APMC sparsifier $\allpairs=(V_{\ap},E_{\ap},w_{\ap})$ of the input $G$,
which does admit efficient minimum $s,t$-cut algorithms. 

The APMC sparsifier is constructed by applying the star transform
to every super-vertex in a $(1/6,2n)$-friendly cut sparsifier $\friendly$ of $G$. 
Let us outline why this process yields an APMC sparsifier.
Fix $s,t\in V$, and recall that applying the star transform on $u$
decomposes $u$ into its constituent vertices,
and adds a new proxy vertex $u'$ connected to each vertex $v\in u$
with an edge of weight $|E(v,u\setminus v)|$.
Consider some minimal minimum $s,t$-cut $S\subseteq V$ in $G$ such that $s\in S$, and denote the super-vertex in $\friendly$ that contains $s$ by $u$.
If $S$ is $1/6$-friendly then it is preserved in $\friendly$.
Otherwise, by \Cref{lemma:structure-min-st-cut} we can write $S=X\cup \set{s}$,
where $X$ is a $1/6$-friendly cut with at most $2n$ edges. 
Now, we have $\mintcut_G(A\cup B) = \mintcut_G(A) + \mintcut_G(B) - 2\cdot |E(A,B)|$
whenever $A,B\subseteq V$ are disjoint,
and thus we can write $\mintcut_G(S)$ as 
\begin{align}
    &\mintcut_G(X\cup\set{s})
    \nonumber
    = 
    \mintcut_G(X)
    -2\cdot|E(X,s)|
    +\deg(s)
    \\
    \nonumber
    &= 
    \mintcut_G(X)
    -2\cdot|E(X,s)|
    + |E(s,u\setminus s)|
    + |E(s,X)|
    + |E(s,V\setminus (u\cup X))|
    \\
    \label{eq:cut-value-decomposition}
    &=
    \mintcut_G(X)
    -|E(X,s)|
    + 
    |E(s,u\setminus s)|
    + |E(s,V\setminus (u\cup X))|
    . 
\end{align}
Next, let us examine the corresponding cut $S=X\cup \set{s}$ in $\allpairs$,
but to be precise, we must assign each proxy vertices to some side of the cut.
For a super-vertex $y\in \friendly$ such that $y\subseteq S$ or $y\subseteq V\setminus S$,
we assign its proxy vertex $y'$ to the same side as (all of) $y$;
this way, none of the edges incident to $y'$ cross the cut.
Since $X$ is $1/6$-friendly,
this handles all the super-vertices in $\friendly$ except (possibly) $u$.
We also handle $u'$ by assigning it to the side opposite to $s$. 
Denote the corresponding cut value by $\mintcut_{\allpairs}(X\cup\set{s})$,
which has a slight abuse of notation because some proxy vertices are omitted.
We can write it, similarly to \eqref{eq:cut-value-decomposition}
but replacing $E(s,u\setminus s)$ with the edge $(s,u')$, 
as
\ifstoc
\else 
\begin{align*}
    &\mintcut_{\allpairs}(X\cup\set{s})
    =
    \mintcut_{\allpairs}(X)
    -w_{\ap}(X,s)
    +w_{\ap}(s,u')
    +w_{\ap}(s,V\setminus (u\cup X))
    ,
\end{align*}
\fi
where $w_{\ap}(\cdot,\cdot)$ is the total weight of the edges in $E_{\ap}(\cdot,\cdot)$.
The cut $X$ is preserved in $\friendly$ since it is $1/6$-friendly,
and hence $\mintcut_{\allpairs}(X)=\mintcut_G(X)$.
Furthermore, $w_{\ap}(X,s) = |E(X,s)|$
and $w_{\ap}(s,V\setminus (u\cup X)) = |E(s,V\setminus (u\cup X))|$
because they correspond to edges outside the super-vertices of $\friendly$.
Finally, by the definition of the star transform
$w_{\ap}(s,u') = |E(s,u\setminus s)|$.
Altogether, $\mintcut_{\allpairs}(X\cup\set{s})=\mintcut_G(X\cup\set{s})$.
We formalize this argument in
\Cref{claim:star-transform-yields-apmc-sparsifier} below;
the main difficulty is to show that $\allpairs$ does not have any ``new'' minimum $s,t$-cut,
e.g., by assigning the proxy vertices differently. 
We then address the running time in
\Cref{claim:summary-constructing-all-pairs-minimum-cut-sparsifier} below.
Putting the claims together immediately proves \Cref{theorem:apmc-sparsifier-construction}. 

\begin{claim}
\label{claim:star-transform-yields-apmc-sparsifier}
Let $\friendly$ be a $(1/6,2n)$-friendly cut sparsifier of
an unweighted graph $G=(V,E)$ on $n$ vertices.
Then performing a star transform on every super-vertex in $\friendly$
yields a graph $\allpairs=(V_{\ap},E_{\ap},w_{\ap})$
that is an APMC sparsifier of $G$ with $|E_{\ap}| \le |E_{\friendly}|+n$.
\end{claim}

\begin{claim}
    \label{claim:summary-constructing-all-pairs-minimum-cut-sparsifier}
    Given a $(1/6,2n)$-friendly cut sparsifier $\friendly$ of an unknown unweighted graph $G=(V,E)$ on $n$ vertices, and the degree of every vertex in $G$,
    one can perform a star transform on every super-vertex in $\friendly$
    in time $O(|E_H|+n)$.
\end{claim}

\begin{proof}[Proof of \Cref{claim:star-transform-yields-apmc-sparsifier}]
Fix a minimal minimum $s,t$-cut $S\subseteq V$ in $G$,
let $U$ be the set of super-vertices in $\friendly$,
and let $U'=\set{u' \mid u\in U}$ be the set of their corresponding proxy vertices in $\allpairs$.
    For every cut $C\subseteq V$ let $C_{\ap}=A\cup C$ where $A=\arg\min_{T\subseteq U'} \mintcut_{\allpairs}(T\cup C)$ (breaking ties arbitrarily), which is the minimum cut corresponding to $C$ in $\allpairs$.
    We begin by showing that for any other $s,t$-cut $S'\subseteq V$ we have $\mintcut_{\allpairs}(S_{\ap}')\ge \mintcut_G(S)$, and then consider $S_{\ap}$ itself and show that $\mintcut_G(S)=\mintcut_{\allpairs}(S_{\ap})$.
    Notice that we only need to consider cuts of the form $C_{\ap}$ for some $C\subseteq V$,
    since they minimize over all cuts in $\allpairs$
    that assign the vertices of $V$ to two sides in the same way ($C$ and $V\setminus C$).
    Combining the two results,
    we obtain that $S_{\ap}$ is a minimum $s,t$-cut in $\allpairs$
    and its value is equal to the minimum $s,t$-cut value in $G$.
    Therefore, $\allpairs$ is an APMC sparsifier of $G$.

    Let $S'\subseteq V$ be some $s,t$-cut in $G$, $U$ be the super-vertices of $\friendly$, and for each $u\in U$ denote $S_u'=(S'\cap u)$.
    We can write the value of the cut $A$ in $G$ as
    \begin{equation*}
        \mintcut_G(S')
        = \sum_{u\in U} 
        |E(S_{u},u\setminus S')|
        + |E(S_{u},V\setminus (S'\cup u))| 
        ,
    \end{equation*}
    by partitioning the edges into those internal to super-vertices and those external to them.
    The term $|E(S_{u}',V\setminus (S'\cup u))|$ is preserved in $\friendly$ and in $\allpairs$
    because it comprises of edges external to super-vertices.
    
    Now, fix some super-vertex $u\in U$ and examine the term corresponding to the internal edges $|E(S_{u}',u\setminus S')|$.
    In the induced graph $\allpairs[u\cup\set{u'}]$, all edges are of the form $(v,u')$ for $v\in u$.
    If $u'\in S_{\ap}'$, the corresponding term is the weight of all the edges between $u'$ and $u\setminus S_{u}'$, given by $w_{\ap}(u\setminus S_{u}', u')=\sum_{v\in u\setminus S_{u}'}|E(v,u\setminus v)|\ge|E(S_{u}',u\setminus S')|$, where the equality is by the properties of the star transform.
    Similarly, if $u'\not\in S_{\ap}'$, the corresponding term is $w_{\ap}(S_{u}',u')=\sum_{v\in S_{u}'}|E(v,u\setminus v)|\ge |E(S_{u}',u\setminus S')|$, where again the equality is by the properties of the star transform.
    Therefore, we find that $\mintcut_{\allpairs}(S_{\ap}')\ge \mintcut_G(S')\ge \mintcut_G(S)$, where the last inequality is by the minimality of $S$.

    We now turn to prove that $\mintcut_G(S)=\mintcut_{\allpairs}(S_{\ap})$.
    If the cut $S$ is $1/6$-friendly then it is preserved in $\friendly$ and hence also in $\allpairs$.
    Otherwise, either $s$ or $t$ is the unique vertex with friendliness ratio $\le 1/6$ by \Cref{claim:at-most-one-unfriendly}, and assume without loss of generality that it is $s$ and that $s\in S$.
    By \Cref{lemma:structure-min-st-cut}, $S=X\cup\set{s}$ where $X$ is a $1/6$-friendly cut with at most $2n$ edges.
    We argue that $\mintcut_{\allpairs}(X_{\ap}\cup\set{s})= \mintcut_G(S)$ which yields the required result.

    Examine the cut $X_{\ap}\cup\set{s}$ in $\allpairs$ and let $u$ be the super-vertex in $\friendly$ that contains $s$.
    Recall that similarly to \eqref{eq:cut-value-decomposition},
    \ifstoc
\else
        \begin{align*}
        \mintcut_G(X\cup\set{s})
        &=\mintcut_G(X)
        -|E(s,X)|
        +|E(s,V\setminus (u\cup X))|
        +|E(s,u\setminus\set{s})|
        \\
        \mintcut_{\allpairs}(X_{\ap}\cup\set{s})
        &=\mintcut_{\allpairs}(X_{\ap})
        -w_{\ap}(s,X)
        +w_{\ap}(s,V\setminus (u\cup X))
        +w_{\ap}(s,u')
        ,
    \end{align*}

    \fi 
    where in the last equation we used that the proxy vertex $u'$ is on the opposite side of the cut from $s$.
    Observe that the proxy vertex $u'$ is not in $X_{\ap}$ since a friendly cut cannot partition any super-vertex in $\friendly$ and $s\not\in X$.
    Hence, $u\cap X_{\ap}=\emptyset$ and therefore including $u'$ in $X_{\ap}$ would increase its value, and we conclude that $u'\not\in X_{\ap}\cup\set{s}$.
    We next prove that these two quantities above are equal,
    by showing that each term is equal to its corresponding term.    

    We begin by showing that $\mintcut_{\allpairs}(X_{\ap})=\mintcut_G(X)$.
    Recall that $X$ is a $1/6$-friendly cut in $G$ and $\mintcut_G(X)\le 2\deg(s)\le 2n$.
    Hence, no edge of $E(X,V\setminus X)$ is contracted in $\friendly$ (and in $\allpairs$).
    In addition, if a cut $A\subseteq V$ includes all the vertices inside a super-vertex $u$, then $S_{\ap}'$ must also include $u'$, and we conclude $\mintcut_{\allpairs}(X_{\ap})=\mintcut_G(X)$.
    Since $E(s,X)$ and $E(s,V\setminus u)$ are sets of edges external to super-vertices of $\friendly$,
    they are also preserved in $\allpairs$.
    Finally, by the star transform operation we obtain $w_{\ap}(s,u')=|E(s,u\setminus s)|$, and $\mintcut_{\allpairs}(X\cup\set{s})=\mintcut_G(X\cup\set{s})$ as required.

    To conclude the proof we show the bound on the number of edges in $\allpairs$.
    Every edge of $E_{\ap}$ that was not added during the star transform operation was already in $E_{\friendly}$.
    Furthermore, for every super-vertex $u$ in $\friendly$
    we add exactly $|u|$ new edges incident to the proxy vertex $u'$.
    Since $\sum_{u} |u| \le n$ we find that $|E_{\ap}|\le |E_{\friendly}|+n$.
\end{proof}

\begin{proof}[Proof of \Cref{claim:summary-constructing-all-pairs-minimum-cut-sparsifier}]
For every super-vertex $u\in \friendly$ and vertex $v\in u$ compute $|E(v,u)|=\deg(v)-E(v,V\setminus u)$,
which altogether takes $O(|E_{\friendly}|)$ time by iterating over each edge in $\friendly$ once.
Performing the star transform only requires knowing
$|E(v,u\setminus v)|=|E(v,u)|$ for every $v\in u$.
It takes $O(|E_{\friendly}|)$ time to reconnect the edges to their original endpoints
and an additional $O(n)$ time to add the edges incident to proxy vertices.
\end{proof}

\section{Applications}
\label{sec:applications}
In this section we leverage the structural results
from \Cref{sec:technical-overview} (proved in \Cref{sec:structural-proof})
to construct an $\sparsifiername$ of an input graph $G$,
in order to prove our main theorems for the cut-query, streaming, and fully-dynamic models.
As seen above, the main ingredient in constructing an APMC sparsifier is a friendly cut sparsifier.
In the cut-query model, one can construct an $(\alpha,w)$-friendly cut sparsifier using the following lemma from \cite{KK25b},
and putting it together with \Cref{theorem:apmc-sparsifier-construction},
we immediately obtain \Cref{theorem:main-result-cut-query}.

\begin{lemma} [Lemma 1.7 of \cite{KK25b}]
\label{lemma:friendly-cut-sparsifier-cut-query}
    Given cut-query access to an unweighted graph $G$ on $n$ vertices
    and parameters $\alpha$ and $w$, 
    one can recover all the edges (along with their endpoints in $G$) of a friendly $(\alpha,w)$-cut sparsifier of $G$
    using $\tO(\alpha^{-1}n\sqrt{w})$ cut queries.
    The algorithm is randomized and succeeds with probability $1-1/\poly(n)$.
\end{lemma}

\begin{proof}[Proof of \Cref{theorem:main-result-cut-query}]
Begin by constructing a $(1/6,2n)$-friendly cut sparsifier $\friendly$ of $G$ using \Cref{lemma:friendly-cut-sparsifier-cut-query} with $\alpha=1/6$ and $w=2n$.
In addition, recover the degree of each vertex $v\in V$ by querying the cut $\set{v}$
and then apply \Cref{theorem:apmc-sparsifier-construction}
to construct an APMC sparsifier $\allpairs$ of $G$.     
Finally, build a Gomory-Hu tree of $\allpairs$ using some offline algorithm.
The query complexity is dominated by the construction of $\friendly$, which takes $\tO(n^{3/2})$ cut queries by \Cref{lemma:friendly-cut-sparsifier-cut-query}.
\end{proof}

Unfortunately, in the streaming and fully-dynamic models,
efficient algorithms for constructing a friendly cut sparsifier are not known,
and we provide the first such algorithms,
based on the randomized construction of \cite{AKT22b}.%
\footnote{There exists also a deterministic construction with slightly worse guarantees.}
The heart of the algorithm is to recursively apply
an expander decomposition procedure on a contraction of the graph,
shave some vertices from each expander and then contract the shaved expanders.
We begin by defining the notion of expander decomposition.

\begin{definition}
    A graph $G$ is a said to be a \emph{$\phi$-expander} if 
    \begin{equation*}
        \forall S\subseteq V,
        \qquad
        \Phi_G(S) 
        \eqdef
        \frac{\mintcut_G(S)}{\min(\vol{S},\vol{V\setminus S})} 
        \ge \phi
        ,
    \end{equation*}
where $\vol{S} \eqdef \sum_{v\in S} \deg(v)$ is called the volume of $S$.
An \emph{$(\epsilon,\epsilon/\phi)$-expander decomposition} of a (multi)graph $G$
is a partition of $V$ into clusters $V=V_1\cup\ldots\cup V_k$
such that the number of edges between clusters is at most $\epsilon \cdot|E|$,
and every induced subgraph $G[V_i]$, for $i\in [k]$, is an $(\epsilon/\phi)$-expander. 
\end{definition}
To fit our needs, we need to strengthen the friendly cut sparsifier construction of \cite{AKT22b}. We make a few modifications.
The first one is to show that the construction succeeds
even if we allow vertices to be shaved when they are within factor $\beta\in(0,1)$ from the shaving conditions;
this is critical in the streaming model,
where one cannot store the entire graph and must rely on sparsifiers to approximate those conditions.
The second modifications is to show that the construction is robust to deletions and insertions.
The sparsifier $\friendly$ is called \emph{$k$-robust}
if given a sequence of $k'$ edge insertions $I=\set{e_1,\ldots,e_{k'}}$
and $k''$ edge deletions $D=\set{e_1,\ldots,e_{k''}}$
such that $k'+k'' \leq k$,
the graph $(\friendly\cup I)\setminus D$ is a $(2\alpha,w)$-friendly cut sparsifier of $(G\cup I)\setminus D$.
We show that the construction of \cite{AKT22b} is in fact $k$-robust for $k=O(\sqrt{w})$.
Finally, the proof of \cite{AKT22b} contains a flaw, which we fix in our construction.
Our next lemma describes an algorithm for constructing a robust friendly cut sparsifier;
its proof appears in \Cref{sec:clusters}, 
and the algorithm itself in \Cref{algorithm:friendly-cut-sparsifier}.

\begin{lemma}
\label{lemma:friendly-cut-sparsifier-AKT}
Suppose one has access to a procedure that,
given a contracted multigraph and $\epsilon\in(0,1)$,
computes for it an $(\epsilon,\epsilon/\phi)$-expander decomposition.
Then, given an unweighted graph $G$ on $n$ vertices,
parameters $\alpha\in(0,1/4),w\ge1$ and slack $\beta\in(0,1)$,
one can compute the super-vertices of
a $\sqrt{w}$-robust $(\alpha,w)$-friendly cut sparsifier $\friendly$ of $G$,
that has $\tO((\beta\alpha)^{-1}n\phi\sqrt{w})$ edges.
This algorithm makes $\polylog(n)$ calls to the expander decomposition procedure on contractions of $G$,
and spends $O(m)$ time outside these calls.%
\footnote{The slack parameter $\beta$ applies in models where one cannot access $G$ directly, but rather via cut sparsifiers. One example is the streaming model, where storing the entire graph is not feasible.}
\end{lemma}

\paragraph*{Streaming Model.}
The main technical challenge in constructing a friendly cut sparsifier in the streaming model
is to apply an expander decomposition recursively
while making only a single pass over the data stream.
Recent work \cite{FKM23,CKMM24} has shown how to construct a single expander decomposition in a single pass.
Using standard techniques, we extend these results
to graphs $G'$ that are contractions of the input graph $G$,
and obtain the following generalization of \cite{FKM23},
whose proof appears in 
\ifstoc
\else
\Cref{sec:dynamic-streams}.
\fi

\begin{theorem}
    \label{theorem:expander-decomposition-streaming-FKM23}
    Given a dynamic stream of an unweighted graph $G$ on $n$ vertices, and parameter $\epsilon\in(0,1)$, one can compute an $(\epsilon,\Omega(\epsilon/\log n))$-expander decomposition of a contraction of $G$ using $\tO(\epsilon^{-2}n)$ space in one pass.
\end{theorem}

Using this theorem together with \Cref{lemma:friendly-cut-sparsifier-AKT},
we can construct a friendly cut sparsifier in the streaming model,
as stated in the next lemma, whose proof appears in 
\ifstoc
\else
\Cref{sec:dynamic-streams}.
\fi
The space requirement of this algorithm does not depend on $w$,
because it only outputs the super-vertices of the friendly cut sparsifier (without its edges).

\begin{lemma}
\label{lemma:friendly-cut-sparsifier-streaming}
Given a graph $G$ presented as a dynamic stream and parameters $\alpha\in(0,1/2),w\ge1$,
on can output the super-vertices of an $(\alpha,w)$-friendly cut sparsifier $\friendly$ that has $|E_{\friendly}|=\tO(\alpha^{-1}n\sqrt{w})$ edges,
using $\tO(\alpha^{-2}n)$ space in one pass.
\end{lemma}

\begin{proof}[Proof of \Cref{theorem:main-result-streaming}]
    In the first pass over the stream, find the super-vertices of a $(1/6,2n)$-friendly cut sparsifier $\friendly$ of $G$ using \Cref{lemma:friendly-cut-sparsifier-streaming}.
    In the second pass, recover the edges of $\friendly$ (along with original endpoints in $G$)
    and the degree of every vertex $v\in V$,
    and use them to construct an APMC sparsifier $\allpairs$ using \Cref{theorem:apmc-sparsifier-construction}.
    Finally, compute the minimum $s,t$-cuts in $\allpairs$ using some offline algorithm.
    The space complexity of the algorithm is dominated by storing the edges of $\friendly$, of which there are at most $\tO(n^{3/2})$ by \Cref{lemma:friendly-cut-sparsifier-AKT}.
\end{proof}

\paragraph*{Fully Dynamic.} 
In the fully dynamic model, we leverage the fact that our friendly cut sparsifier construction is robust to edge updates.
This allows us to use a static algorithm to construct a $(\polylog n, 2n)$-friendly cut sparsifier $H$ of $G$, which is then valid for $\sqrt{n}$ edge updates, during which we can construct a new friendly cut sparsifier $\friendly'$ in the background.
Then, while the friendly cut sparsifier is valid every edge update requires $O(1)$ updates to the $\sparsifiername$ $F$.
To obtain a deterministic algorithm, we employ a deterministic expander-decomposition procedure from \cite{LS21}.%
\footnote{In \cite{AKT22b} there is a different algorithm for deterministic construction of a friendly cut sparsifier.
  It leverages a stronger notion of expander decomposition,
  and employs a single expander decomposition instead a recursive one.
  To simplify our presentation, we use the same meta-algorithm
  for both the deterministic and randomized versions.
}

\begin{theorem}[Corollary 2.5 of \cite{LS21}]
    \label{theorem:expander-decomposition-deterministic-LS21}
    There exists a deterministic algorithm that, given a weighted graph $G=(V,E,w)$ with $m$ edges with weight bounded in $1\le w_e \le U$ for all $e\in E$, and parameters $\epsilon\in(0,1)$, $r\ge 1$, constructs an $(\epsilon,\epsilon/(\log^{O(r^4)}m \log U))$-expander decomposition of $G$ in time $m(mU)^{O(1/r)}\log (mU)$.
\end{theorem}

We employ \Cref{theorem:expander-decomposition-deterministic-LS21} with $r=\log\log n$ and $U=O(n^2)$ to obtain an $(\epsilon,\epsilon/n^{o(1)})$-expander decomposition in time $m^{1+o(1)}$.
Plugging this into \Cref{lemma:friendly-cut-sparsifier-AKT} yields a deterministic algorithm for constructing a friendly cut sparsifier in time $m^{1+o(1)}$.
Finally, to construct the Gomory-Hu tree of $G$ we rebuild it from scratch on the $\sparsifiername$ after every update using the deterministic $m^{1+o(1)}$-time algorithm of \cite{AKL+25}.

\begin{theorem}
    \label{theorem:gomory-hu-tree-static-AKL+25}
    Given an undirected, weighted graph $G=(V,E,w)$ with polynomially bounded weights, a Gomory-Hu tree of $G$ can be constructed in $m^{1+o(1)}$ time.
\end{theorem}

\begin{proof}[Proof of \Cref{theorem:main-result-fully-dynamic}]
The fully-dynamic algorithm is as follows.
Throughout the updates, use the method explained further below to maintain
a $(1/6,2n)$-friendly cut sparsifier $\friendly$ of the current graph $G$ with $|E_{\friendly}| = n^{3/2+o(1)}$ and also the list of degrees in $G$.
In addition, construct from it using \Cref{theorem:apmc-sparsifier-construction} 
an $\sparsifiername$ $\allpairs$ of $G$ with $O(|E_{\friendly}|+n) = n^{3/2+o(1)}$ edges.
Finally, build for $\allpairs$ a Gomory-Hu tree using \Cref{theorem:gomory-hu-tree-static-AKL+25}.
The time for constructing the APMC sparsifier and the Gomory-Hu tree
is thus $(O(|E_{\friendly}|+n))^{1+o(1)} = n^{3/2+o(1)}$ 
by the above bound on $|E_{\friendly}|$. 

It remains to show how to maintain a $(1/6,2n)$-friendly cut sparsifier $\friendly$ using $n^{3/2+o(1)}$ worst-case update time.
We divide the update sequence into epochs of $\sqrt{n}$ updates, and assume that at the beginning of each epoch we have a valid $(1/12,2n)$-friendly cut sparsifier $\friendly$ of the current graph $G$.
During the epoch, update the edges of $\friendly$ according to edge updates to $G$,
namely, whenever an edge is inserted/deleted into $G$,
insert/delete that edge into $\friendly$ (and update the degree list).
Since the friendly cut sparsifier is $\sqrt{n}$-robust,
it remains a $(1/6,2n)$-friendly cut sparsifier after $\sqrt{n}$ edge updates.
Therefore, throughout each epoch we can maintain $\friendly$ using $O(1)$ time per update.

We now explain how to construct $\friendly$ for the next epoch.
For the first epoch, this is trivial because the graph $G$ is empty
and thus $G$ itself is a friendly cut sparsifier.
During the last $n^{1/2}/\log n$ updates of each epoch,
construct in the background a new $(1/24,2n)$-friendly cut sparsifier $\friendly'$ of the current graph $G$.
The total time for constructing $\friendly'$ is $m^{1+o(1)}\le n^{2+o(1)}$,
where $m$ is the number of edges in $G$ at the beginning of the construction,
by \Cref{lemma:friendly-cut-sparsifier-AKT} and the complexity of the expander-decomposition procedure we use (\Cref{theorem:expander-decomposition-deterministic-LS21} with $r=\log\log n$ and $U=O(n^2)$).
Since this construction is spread over $\sqrt{n}/\log n$ updates,
it takes $n^{2+o(1)}/(\sqrt{n}/\log n) = n^{3/2+o(1)}$ worst-case time per update.
In addition, maintain the updates to $G$ that occur during this construction
in a buffer, and apply them to $\friendly'$ once it is constructed,
which takes $O(\sqrt{n}/\log n)$ time as each edge update takes $O(1)$ time. 
By the robustness property of $\friendly'$,
it remains a $(1/12,2n)$-friendly cut sparsifier of $G$ after applying these updates.
Overall, the algorithm uses $n^{3/2+o(1)}$ worst-case update time to maintain $\friendly$.
Finally, our parameter choice implies $\phi=n^{o(1)},\alpha=1/6,w=2n$ and $\beta=1$,
and thus by \Cref{lemma:friendly-cut-sparsifier-AKT}
we have $|E_{\friendly}| = n^{3/2+o(1)}$.
\end{proof}

\section{Constructing a Friendly Cut Sparsifier}
\label{sec:clusters}
In this section we present a generic algorithm for constructing a friendly cut sparsifier, based on the work of \cite{AKT22b}.
The algorithm is based on recursively applying an expander decomposition procedure on contractions of the graph, shaving some vertices from each expander and then contracting the shaved expanders.
A vertex $v$ is shaved from an expander if its degree is small compared to either the number of its incident edges exiting the expander or to a given threshold.
We note that the original proof of the correctness of the algorithm had a flaw, which we fix later in the section.
Furthermore, we strengthen the construction in two senses.
The first, is that we allow for slack in the shaving condition, i.e. we do not contract some vertices with degree up to $\beta^{-1}$ times as large as the threshold for some $\beta\in (0,1)$.
This is important in the streaming setting where one cannot directly access all graph edges, and consequently cannot verify contraction conditions exactly.

The second, is showing that the construction is $\alpha^{-1}\sqrt{w}$-robust, i.e. it continues to be a friendly cut sparsifier after $\alpha^{-1}\sqrt{w}$ edge insertions and deletions (and appropriately modifying the sparsifier).
To do so we first prove that the construction also preserves all cuts $S\subseteq V$ such that every vertex $v\in S$ with degree at least $10\alpha^{-1}\sqrt{w}$ has friendliness ratio at most $\alpha$ (but vertices with smaller degree may have friendliness ratio less than $\alpha$).
This is stronger than the original condition, which required all vertices in $S$ to have friendliness ratio at most $\alpha$.
The modified algorithm for constructing the friendly cut sparsifier is given in \Cref{algorithm:friendly-cut-sparsifier}.

\begin{algorithm}
    \caption{Friendly Cut Sparsifier Construction}
    \label{algorithm:friendly-cut-sparsifier}
    \begin{algorithmic}[1]
        \State \textbf{Input:} An unweighted graph $G=(V,E)$, parameters $w\ge1,\alpha\in(0,1/2),\beta\in(0,1)$
        \State \textbf{Output:} A robust $(\alpha,w)$-friendly cut sparsifier
        \Procedure{Friendly-Cut-Sparsifier}{$G,\alpha,w$}
        \State $\epsilon\gets 0.01 (\alpha\beta), j\gets 1$
        \State $G_0=(V_0,E_0) \gets G$
        \While{$w_j\ge w$}
            \State $w_j \gets 4^{-j} (m/n)^2$
            \State add $(\epsilon/\phi)^{-1}\sqrt{w_j}$ self loops to every $v\in V_{j-1}$ \Comment{$\phi$ is determined by the expander decomposition procedure}
            \State $W_1,\ldots,W_k \gets (\epsilon,\epsilon/\phi)$-expander decomposition of $G_{j-1}$
            \For{$i\in [k]$}
                \State $R_i \gets \Big\{ v\in W_i \mid  d_{G_{j-1}}(v)<\max \{ 10\alpha^{-1}\sqrt{w_j} |v|, 4\alpha^{-1}\cdot |E_{j-1}(\set{v}, V_{j-1}\setminus W_i)|\} \Big\}$%
                \label{lst:line:contraction-criteria}
                \State $R_i' \gets \Big\{ v\in W_i \mid  \beta d_{G_{j-1}}(v)<\max \{ 10\alpha^{-1}\sqrt{w_j} |v|, 4\alpha^{-1}\cdot |E_{j-1}(\set{v}, V_{j-1}\setminus W_i)|\} \Big\}$
                \State $R_i \gets R_i \cup \rho$ where $\rho\subseteq R_i'$ \Comment{arbitrary subset of $R_i'$}
                \State $W_i' \gets W_i \setminus R_i$
            \EndFor
            \State $G_j =(V_j,E_j) \gets$ contract in $G_{j-1}$ each $W_i'$ into a single vertex $v_i$ \Comment{keeping parallel edges and discarding self loops}
            \State $j \gets j+1$
        \EndWhile
        \State \Return $G_{j-1},\set{W_i}$.
        \EndProcedure
    \end{algorithmic}
\end{algorithm}

The proof of \Cref{lemma:friendly-cut-sparsifier-AKT} is by induction.
For the base case we trivially have that $G_0=G$ is a friendly cut sparsifier of $G$, and the following two claims complete the induction step.
Both claims are based on analogous claims in \cite{AKT22b}, however the proofs are slightly different due to the changes in the algorithm.
Furthermore, our proof fixes the aforementioned flaw in the original proof of \cite{AKT22b}.
Finally, the runtime bound follows as the expander decomposition is computed $O(\log (m/n))=O(\log n)$ times, and the other steps in each iteration can be implemented in linear time in the number of edges.
\begin{claim}[Correctness]
    \label{claim:correctness-of-iterative-friendly-cut-sparsifier}
    Let $S\subseteq V$ be a cut with at most $w$ edges such that every $v\in V$ with $\deg(v)\ge 10\alpha^{-1}\sqrt{w}$ has friendliness ratio at least $\alpha$.
    Then for every iteration $j$ of \Cref{algorithm:friendly-cut-sparsifier}, we have $\mintcut_{G_j}(S)=\mintcut_G(S)$.
\end{claim}
\begin{claim}[Size]
    \label{claim:friendly-sparsifier-size}
    The number of edges in $G_j$ when running \Cref{algorithm:friendly-cut-sparsifier} with slack $\beta$ and $\phi\ge \log n$ is at most $100\phi(\alpha\beta)^{-1} n\sqrt{w_j}$.
\end{claim}
Finally, the robustness of the sparsifier is due to the following lemma, whose proof is provided at the end of the section.
\begin{lemma}
    \label{lemma:robustness-of-friendly-cut-sparsifier}
    Let $\friendly$ be the output of \Cref{algorithm:friendly-cut-sparsifier} with $\alpha<1/4$, $D=\set{e_1,\ldots,e_a}$ a sequence of $a$ edge deletion and  $I=\set{e_1,\ldots,e_b}$ a sequence of edge insertions such that $\sqrt{w}\ge a+b$.
    Then, the graph $(\friendly\cup I)\setminus D$ is a $(2\alpha,w)$-friendly cut sparsifier of $(G\cup I)\setminus D$.
\end{lemma}

\subsection{Proof of Correctness}
In this section we give a corrected proof of Claim 2.3 from \cite{AKT22b}, which states that the output of \Cref{algorithm:friendly-cut-sparsifier} is a friendly cut sparsifier.
Note that our proof is actually stronger as it shows that the sparsifier preserves even cuts that are not $\alpha$-friendly, as long as all the vertices with low friendliness ratio have low degree.

We now describe the flaw in the original proof of \cite{AKT22b} and how we fix it.
Fix a cut $S\subseteq V$ such that $\mintcut_G(S)\le w$ and some super-vertex $x\in S$ (recall that $x$ is in the contracted graph $G_{j-1}$).
Using this notation, the proof attempts to show that $|E_{j-1}(x,S)|\ge \alpha\deg_{G_{j-1}}(x)$;
i.e. that the number edges incident to $x$ that exit $S$ is at least an $\alpha$-fraction of the degree of $x$ in $G_{j-1}$.
To do so, the faulty proof uses the following bound
\begin{equation}
    \label{eq:incorrect-edge-friendliness-lower-bound}
    |E_{j-1}(x,S)| 
    \ge \sum_{u\in x} |E(u,S)|
    ,
\end{equation}
which does not hold in general since some edges in $E(u,S)$ may be incident to other vertices in $x$, and thus are not counted in $|E_{j-1}(x,S)|$.
We show that a slightly weaker bound still holds and is sufficient for constructing a friendly cut sparsifier.
\begin{claim}
    \label{claim:edge-friendliness-lower-bound}
    For every super-vertex $x\in V_{j-1}$ such that the degree of every $u\in x$ is at least $10\alpha^{-1}\sqrt{w}$, we have $|E_{j-1}(x,S)|\ge \frac{9\alpha}{10} \deg_{G_{j-1}}(x)$.
\end{claim}
\begin{proof}[Proof of \Cref{claim:edge-friendliness-lower-bound}]
    We split the proof into two cases.
    First, when $d_{G_{j-1}}(x)\ge 2w$ we have that 
    \begin{equation*}
        \frac{|E_{j-1}(x,S)|}{\deg_{G_{j-1}}(x)}
        =
        \frac{\deg_{G_{j-1}}(x)-|E_{j-1}(x,V_{j-1}\setminus S)|}{\deg_{G_{j-1}}(x)}
        \ge
        1-\frac{w}{2w}
        =\frac{1}{2}
        ,
    \end{equation*}
    where the inequality is since $|E_{j-1}(x,V_{j-1}\setminus S)|\le \mintcut_{G_{j-1}}(S)\le w$ by the induction hypothesis.
    Since $\alpha\le 1/2$ we obtain the desired bound.
    In the complementary case, we will show that $\deg_{G_{j-1}}(x)\ge (1-\alpha^2/50)\cdot\vol{x}$.
    Intuitively, this means that most edges incident to vertices in $x$ are external to the super-vertex rather than internal edges.
    This allows us to recover a bound similar to \Cref{eq:incorrect-edge-friendliness-lower-bound}.
    Formally, 
    \begin{align*}
        |E_{j-1}(x,S)|
        &= \sum_{u\in x} |E(u,S)| - |E(u,x)|
        \\
        &= \deg_G(x) - (\vol{x}-\deg_{G_{j-1}}(x))
        \\
        & \ge \sum_{u\in x} \alpha \deg_G(u)- \frac{\alpha^2}{100}\vol{x}
        = (\alpha-\frac{\alpha^2}{50})\vol{x}
        \\
        &\ge \left( \alpha-\frac{\alpha^2}{50} \right) \deg_{G_{j-1}}(x)
        \ge \frac{9\alpha}{10}\deg_{G_{j-1}}(x)
        ,
    \end{align*}
    where the second equality is from $\sum_{u\in x}|E(u,x)|=\vol{x}-\deg_{G_{j-1}}(x)$, the first inequality uses $\sum_{u\in x} |E(u,S)|\ge \alpha\deg_G(u)$ by the friendliness ratio of every vertex with $\deg(v)\ge10\alpha^{-1}\sqrt{n}$ and $\deg_{G_{j-1}}(x)\ge (1-\alpha^2/50)\cdot\vol{x}$.
    The second inequality is since every vertex $u\in x$ has degree at least $10\alpha^{-1}\sqrt{w}$, and hence its friendliness ratio is at least $\alpha$ by the condition on $S$, and the last inequality is since $\alpha\le 1/2$.

    We conclude by showing that $\deg_{G_{j-1}}(x)\ge (1-\alpha^2/100)\cdot\vol{x}$.
    Notice that $\deg_{G_{j-1}}(x)\ge \vol{x}-2\binom{|x|}{2}$ since $G$ is an unweighted graph.
    By line~\ref{lst:line:contraction-criteria} of \Cref{algorithm:friendly-cut-sparsifier}, we have that $\deg(u)\ge 10\alpha^{-1}\sqrt{w}$ for every $u\in x$ and hence $\vol{x}\ge 10\alpha^{-1}\sqrt{w}|x|$.
    Furthermore, by our assumption above $2w \ge \deg_{G_{j-1}}(x)\ge 10\alpha^{-1}\sqrt{w}|x|$, where the last inequality holds as $x$ is included in the shaved cluster $W_i'$.
    Moving sides we find $|x|\le \alpha \sqrt{w}/5$.
    Finally, using the bounds on $\vol{x}$ and $|x|$,
    \begin{equation*}
        \frac{\deg_{G_{j-1}}(x)}{\vol{x}}
        \ge 1-\frac{2\binom{|x|}{2}}{\vol{x}}
        \ge 1-\frac{|x|^2}{10\alpha^{-1}\sqrt{w}|x|}
        \ge 1-\frac{\alpha \sqrt{w}/5}{10\alpha^{-1}\sqrt{w}}
        \ge 1-\frac{\alpha^2}{50}
        .
        \qedhere
    \end{equation*}
\end{proof}
\begin{proof}[Proof of \Cref{claim:correctness-of-iterative-friendly-cut-sparsifier}]
    By the induction hypothesis, $\mintcut_{G_{j-1}}(S)=\mintcut_G(S)$.
    When this assumption holds no vertex in $S$ is contracted with any vertex in $V\setminus S$ in any previous iterations.
    Since this is the case, we will abuse the notation by treating $S$ as a vertex subset of both $G$ and $G_{j-1}$.
    Assume towards contradiction that two super vertices $x,y\in V_{j-1}$ such that $x\in S,y\not\in S$ are contracted in iteration $j$.

    If $x,y$ are contracted then they must be in the same shaved cluster $W_i'$.
    Let $L=W_i\cap S$ and $R=W_i\setminus S$.
    Note that both are sets not empty because $x\in L,y\in R$.
    Since $W_i$ is an $\epsilon/\phi$-expander we have
    \begin{equation}
        \label{eq:friendly-sparsifier-contradiction}
        \min \set{\vol{L},\vol{R}}
        \le
        \frac{|E_{j-1}(L,R)|}{\epsilon/\phi} 
        \le 
        \frac{w \phi}{\epsilon}
        ,
    \end{equation}
    where the last inequality follows since $|E_{j-1}(L,R)|\le \mintcut_{G_{j-1}}(S)=\mintcut_G(S)\le w$ by the induction hypothesis.
    To proceed, we lower bound $\min \set{\vol{L},\vol{R}}$ to show a contradiction.
    We focus on lower bounding $\vol{L}$, and the argument for $\vol{R}$ is symmetric.
    Notice that since $x$ was not shaved, then $|E_{j-1}(x,V_{j-1}\setminus W_i)|\le \alpha \deg_{G_{j-1}}(x)/4$ by line~\ref{lst:line:contraction-criteria} of \Cref{algorithm:friendly-cut-sparsifier}.
    Furthermore, every vertex $u\in x$ has degree at least $10\alpha^{-1}\sqrt{w_j}$ since it was contracted in some prior iteration.
    Therefore, by \Cref{claim:edge-friendliness-lower-bound} we have $|E_{j-1}(x,S)|\ge \frac{9\alpha}{10} \deg_{G_{j-1}}(x)$,
    Combining these two bounds, we find that
    \begin{equation*}
        |E_{j-1}(x,L)|
        \ge \left( \frac{9}{10}\alpha-\frac{\alpha}{4} \right)\deg_{G_{j-1}}(x)
        \ge \frac{13}{2}\sqrt{w_j}|x|
        ,
    \end{equation*}
    where the first inequality is since $L$ is $S$ limited to $W_i$ and the last inequality follows since $x$ was not shaved, and thus $\deg_{G_{j-1}}(x)\ge 10\alpha^{-1}\sqrt{w_j}|x|$ by line~\ref{lst:line:contraction-criteria} of \Cref{algorithm:friendly-cut-sparsifier}.
    Now observe that $|E_{j-1}(x,L)|\le |x|\cdot |L|$ since $G$ is an unweighted graph.
    These two bounds imply that $|L|\ge \frac{13}{2}\sqrt{w_j}$.
    Recalling that $(\epsilon/\phi)^{-1}\sqrt{w_j}$ self loops are added to every vertex in $G_{j-1}$ at the beginning of iteration $j$, we have that $\vol{L}\ge \frac{13}{2}\frac{1}{\epsilon/\phi}w_j$.
    Thus, (recalling that a symmetric argument holds for $\vol{R}$) we arrive at a contradiction to \Cref{eq:friendly-sparsifier-contradiction} since $\min\set{\vol{L},\vol{R}}\ge \frac{13}{2}\frac{1}{\epsilon/\phi}w_j > \frac{w\phi}{\epsilon}$.
\end{proof}

\subsection{Proof of Size Guarantee}
In this section we prove \Cref{claim:friendly-sparsifier-size}, which is similar to Claim 2.4 in \cite{AKT22b}.
The main difference between the proofs is the addition of the slack parameter $\beta$ in the contraction criteria.
\begin{proof}[Proof of \Cref{claim:friendly-sparsifier-size}]
  Notice that $G'$ has three types of edges,
  \begin{enumerate}
    \item The outer edges of the expander decomposition, of which there are at most $m_j'\coloneqq\epsilon(|E_{j-1}|+n\cdot(\epsilon/\phi)^{-1}\sqrt{w_{j}})$, where the second term is since $(\epsilon/\phi)^{-1}\sqrt{w_{j}}$ self loops are added to every vertex in $G_{j-1}$ prior to the expander decomposition.
    \item The edges adjacent to vertices shaved because of low degree, i.e. those with $\beta\deg(v) < 10\alpha^{-1}\sqrt{w_j}$, of which there are at most $10(\alpha\beta)^{-1}n\sqrt{w_{j}}$.
    \item Edges that are incident to vertices that were shaved because at least $\alpha \deg(v)/4$ of their edges went outside their expander. 
    For every $v\in V$ let $d_{out}(v)$ be the cardinality of the edge set $E(\set{v},V\setminus W)$, where $W$ is the cluster to which $v$ belongs. 
    Furthermore, let $X\subseteq V$ be the set of the aforementioned shaved vertices, observe that the total number of inter-cluster edges is at most $O(\epsilon(|E_{j-1}|+n\cdot(\epsilon/\phi)^{-1}\sqrt{w_{j-1}})) = \sum_{v\in V}d_{out}(v)$ as explained above. 
    Since for every $x\in X$ we have $\beta^{-1} d_{out}(x) \ge \alpha^{-1} \deg(x)/4$, we find that the number of edges incident to $X$ is at most $4(\alpha\beta)^{-1}m_j'$.
  \end{enumerate}
  Summing up the three types of edges, we find that the total number of edges in $G_j$ is at most
  \begin{align*}
    |E_j| 
    &\le(4(\alpha\beta)^{-1}+1)m_j'+(\alpha\beta)^{-1} n\sqrt{w_j}
    \\
    &= 
    (4(\alpha\beta)^{-1}+1)\cdot \epsilon
    \left( |E_{j-1}|+n\cdot(\epsilon/\phi)^{-1}\sqrt{w_{j}} \right)
    +10(\alpha\beta)^{-1} n\sqrt{w_j}
    \\
    &\le 
    8(\alpha\beta)^{-1}\cdot \epsilon|E_{j-1}|
    + 16\phi(\alpha\beta)^{-1} n\sqrt{w_j}
    ,
  \end{align*}
  where the last inequality is by $\alpha\le 1/2$ and $\phi \ge \log n$.
  Substituting  $|E_{j-1}| \le 100 \phi (\alpha\beta)^{-1}n\sqrt{w_{j-1}}$ by the induction hypothesis, $\epsilon=0.01\alpha\beta$, and $w_{j-1}=4w_j$ we find that
    \begin{equation*}
    |E_j| 
    \le
    16 \phi (\alpha\beta)^{-1}n\sqrt{w_{j}}
    + 16\phi(\alpha\beta)^{-1} n\sqrt{w_j}
    \le 100 \phi (\alpha\beta)^{-1}n\sqrt{w_j}
    \qedhere
  \end{equation*}
\end{proof}
\subsection{Robustness to Edge Updates}
\begin{proof}[Proof of \Cref{lemma:robustness-of-friendly-cut-sparsifier}]
    Denote $G' = (G\cup I)\setminus D$ and $\friendly' = (\friendly\cup I)\setminus D$ and let $S\subseteq V$ be a $2\alpha$-friendly cut in $G'$ with at most $w$ edges.
    Denote the edge set of $G'$ by $E'$
    We will show that $\mintcut_{\friendly'}(S)=\mintcut_{G'}(S)$, notice that this is equivalent to showing that no edge in $E'(S,V\setminus S)$ is contracted in $\friendly'$.
    We prove this by showing that $S$ satisfies the conditions of \Cref{claim:correctness-of-iterative-friendly-cut-sparsifier}, i.e.  every vertex $v\in V$ with $\deg_G(v) \ge 10\alpha^{-1}\sqrt{w}$ has friendliness ratio at least $\alpha$ in $G$ (in regard to $S$).
    By \Cref{claim:correctness-of-iterative-friendly-cut-sparsifier}, this implies that no edge in $E(S,V\setminus S)$ is contracted in $\friendly$, and therefore no edge in $E'(S,V\setminus S)$ is contracted in $\friendly'$ (since updates only add or remove edges, without changing the contraction structure).

    Assume towards contradiction that there exists a vertex $v\in V$  with $\deg_G(v) \ge 10\alpha^{-1}\sqrt{n}$ and friendliness ratio less than $\alpha$ in $G$ but friendliness ratio more than $2\alpha$ in $G'$.
    We now examine how the friendliness ratio of $v$ changes from $G$ to $G'$.
    Notice that the friendliness ratio of $v$ changes the most if all the edge updates are incident to $v$.
    Furthermore, we show that the friendliness ratio changes the most if all edge updates are insertions.
    After a single edge insertion the friendliness ratio of $v$ in $G'$ is at most $1-\frac{\cross{v}{}}{\deg(v)+1}$.
    Similarly, after a single edge deletion the friendliness ratio of $v$ in $G'$ is at most $1-\frac{\cross{v}{}-1}{\deg(v)-1}$.
    Now notice that,
    \begin{equation*}
        1-\frac{\cross{v}{}}{\deg(v)+1}
        >
        1-\frac{\cross{v}{}-1}{\deg(v)-1}
        ,
    \end{equation*}
    follows if $2\cross{v}{}-\deg(v)>1$;
    this holds since the friendliness ratio of $v$ in $G'$ is always less than $1/2$ as $\alpha<1/4$ and $k<\deg(v)/2$.
    Therefore, it remains to analyze the case when all edge updates are insertions.
    The friendliness ratio of $v$ in $G'$ after $k$ edge insertions is at most
    \begin{align*}
        1-\frac{\cross{v}{}}{\deg(v)+k}
        <
        1-\frac{(1-\alpha)\deg(v)}{\deg(v)+\alpha\deg(v)}
        <2\alpha
        ,
    \end{align*}
    where the first inequality is since the friendliness ratio of $v$ in $G$ is less than $\alpha$, and since $k\le\sqrt{w}\le \alpha \deg(v)$ by our assumption on the degree of $v$.
    Therefore, $v$ has friendliness ratio smaller than $2\alpha$ in $G'$, and we arrive at a contradiction.
\end{proof}

\ifstoc
\else

\section{Friendly Cut Sparsifiers in Dynamic Streams}
\label{sec:dynamic-streams}
In this section we provide an algorithm for finding the super-vertices of an $(\alpha,w)$-friendly cut sparsifier in the dynamic streaming model, proving \Cref{lemma:friendly-cut-sparsifier-streaming}.
The construction is an implementation of \Cref{algorithm:friendly-cut-sparsifier} in the dynamic streaming model.

Recall that the main ingredient in \Cref{algorithm:friendly-cut-sparsifier} is performing an expander decomposition on a contraction of the input graph at each iteration of the main loop.
The expander decomposition procedure we leverage is a variant of the algorithm of \cite{FKM23}.
The main technical tool used in the expander decomposition is a $(\delta,\epsilon)$ power cut sparsifier.
\begin{definition}
    A $(\delta,\epsilon)$-cut sparsifier for a (multi)graph $G=(V,E_G)$ is a graph $H=(V,E_H,w_H)$ such that,
    \begin{equation*}
        \forall S\subseteq V,
        \qquad
        (1-\delta)\cdot w_G(E(S,\overline{S})) -\epsilon\cdot \vol{S}
        \le
        w_H(S,\overline{S})
        \le
        (1+\delta)\cdot w_G(E(S,\overline{S})) +\epsilon\cdot \vol{S}
        .
    \end{equation*}
     A distribution $\mathcal{D}$ over graphs $H$ is called a $(\delta,\epsilon,p)$-power cut sparsifier distribution for $G$, if for every partition $\mathcal{C}=\set{C_1,\ldots,C_k}$ of $V$ into disjoint sets, with probability at least $1-p$ over the choice of $H\sim \mathcal{D}$, for all $C\in\mathcal{C}$ we have that $H\{C\}$ is a $(\delta,\epsilon)$-cut sparsifier of $G\{C\}$,
     where $G\{C\}$ is the induced graph on $C$ after adding self loops to every $v\in C$ such that it has the same degree in $G\{C\}$ as in $G$.
     A sample $H$ from $\mathcal{D}$ is called a $(\delta,\epsilon,p)$-power cut sparsifier of $G$.
\end{definition}
One of the main technical results of \cite{FKM23} is that one can sample from a $(\delta,\epsilon,n^{-C})$-power cut sparsifier distribution in a dynamic stream using $\tO(n/(\delta\epsilon))$ space.
Combining this with an expander decomposition tailored to power cut sparsifiers, they obtain the following result.
\begin{theorem}[Theorem 1 of \cite{FKM23}]
    \label{theorem:expander-decomposition-from-power-cut-sparsifier}
    Given $K=O(\epsilon^{-1} \polylog n)$ $(1/(5\log n),\epsilon)$-power cut sparsifiers of a graph $G$ on $n$ vertices for some $\epsilon\ge 0$, one can find an $(\epsilon,\epsilon/\log n)$-expander decomposition of $G$.
\end{theorem}
In order to perform expander decomposition recursively, we need to sample from a $(\delta,\epsilon,n^{-C})$-power cut sparsifier distribution for a contraction of $G$, which is only defined after the stream ends.
The following lemma, whose proof is provided in \Cref{sec:power-cut-sparsifier-contraction}, gives an algorithm for this.
We note that an earlier version of \cite{FKM23} contained a similar result, however it was removed in the most up-to-date version.%
\footnote{The construction was part of a proof for constructing a spanner in the streaming setting that contained an error, and thus was removed. However, the source of the error was elsewhere, and the expander decomposition for a contraction is still valid.}
\begin{lemma}
    \label{lemma:power-cut-sparsifier-contraction}
        Given an unweighted graph $G$ on $n$ vertices presented as a dynamic stream, and a contraction $G'$ of $G$,
        there exists an algorithm using $\tO(n/(\delta\epsilon))$ space that outputs with high probability a sample from a $(\delta,\epsilon,n^{-C})$-power cut sparsifier distribution of $G'$ for any constant $C>0$.
\end{lemma}
The proof of \Cref{lemma:friendly-cut-sparsifier-streaming} also needs a cut sparsifier of $G$ to estimate vertex degrees.
\begin{definition}%
A \emph{$(1\pm\epsilon)$-cut sparsifier} a graph $G$ is a weighted subgraph $H$
that satisfies
  \begin{equation*}
    \forall S\subseteq V,
    \qquad
    (1-\epsilon)\cdot\mintcut_G(S) 
    \leq \mintcut_H(S) 
    \leq (1+\epsilon)\cdot\mintcut_G(S)
    .
  \end{equation*}
\end{definition}
\begin{theorem}[Theorem 3.3 of \cite{AGM12}]
    \label{theorem:cut-sparsifier-streaming}
    Given a graph $G$ presented as a dynamic stream,
    there exists an algorithm that constructs a $(1\pm\epsilon)$-cut sparsifier of $G$ using $\tO(n/\epsilon^2)$ space in one pass.
\end{theorem}
\begin{proof}[Proof of \Cref{lemma:friendly-cut-sparsifier-streaming}]
    Throughout the proof we use $\beta=1/2$ as the slack parameter in \Cref{algorithm:friendly-cut-sparsifier} and fix $\epsilon=0.01(\alpha\beta)\cdot \log^{-3} n$.
    The proof follows by implementing \Cref{algorithm:friendly-cut-sparsifier} in the dynamic streaming model using \Cref{theorem:expander-decomposition-from-power-cut-sparsifier} and \Cref{lemma:power-cut-sparsifier-contraction}.
    Begin by constructing $O((K\cdot \log n)+\log n)$ copies of the data structure of \Cref{lemma:power-cut-sparsifier-contraction} during the stream with $\delta=1/(5\log n)$ and $\epsilon$ as above, where $K=\Theta(\epsilon^{-1} \polylog n)$.
    This is because $K$ power cut sparsifiers are needed for each level of the expander decomposition, and there are $O(\log n)$ levels of recursion.
    In parallel, maintain a quality $(1\pm 1/100)$-cut sparsifier $\cutspar$ of $G$ using \Cref{theorem:cut-sparsifier-streaming}.
    Therefore, the overall space complexity is $\tO(n/\epsilon^{-2})=\tO(n/\alpha^{-2})$.

    We now explain how to implement the main loop of \Cref{algorithm:friendly-cut-sparsifier} with slack parameter $\beta=1/2$.
    Fix some iteration $j$ of the main loop.
    Begin by adding $(\epsilon/\phi(n))^{-1}\sqrt{w_j}$ self loops to every $v\in V_{j-1}$ in every power cut sparsifier for this level.
    Note that this can be implemented by simulating the insertion of these self loops at the end of the stream.
    Then, find a $(\epsilon,\epsilon/\log n)$-expander decomposition $W_1,\ldots,W_k$ of $G_{j-1}$ by using \Cref{theorem:expander-decomposition-from-power-cut-sparsifier}.
    
    The main challenge is finding the shaved clusters $W_i'$.
    Let $\widetilde{\deg}(u)\eqdef\mintcut_{\cutspar}(u)$.
    We now present the modified shaving conditions that use $\widetilde{\deg}(u)$ to estimate the degree of $u$ and a power cut sparsifier to estimate the number of edges between $u$ and $V_{j-1}\setminus W_i'$.
    For each modified criteria we first show that it successfuly shaves every vertex satisfying the original criteria, and then show that any additional vertex is shaved has slack $1/2$.
    The first condition is to shave every super vertex $u\in V_{j-1}$ that satisfies
    \begin{align*}
        \widetilde{\deg}(u)/(1+1/100)<10\alpha^{-1}\sqrt{w_j}|u|
    \end{align*}
    Notice that by the guarantee of the cut sparsifier, if $\deg_{G_{j-1}}(u)<10\alpha^{-1}\sqrt{w_j}|u|$ then 
    \begin{equation*}
        \frac{\widetilde{\deg}(u)}{1+1/100}
        <\frac{(1+1/100)d_{G_{j-1}}(u)}{1+1/100}
        =d_{G_{j-1}}(u)
        <10\alpha^{-1}\sqrt{w_j}|u|
        ,
    \end{equation*}
    and thus the algorithm shaves $u$.
    Furthermore, no vertex of degree at least $10\alpha^{-1}\sqrt{w_j}|v|/(1-1/100)$ is shaved.
    Therefore, this implements the degree condition in line~\ref{lst:line:contraction-criteria} with slack $1/(1+1/100)>1/2$.
    
    For the second condition, we need to estimate the number of edges between $u$ and $V_{j-1}\setminus W_i$.
    Notice that $E_{G_{j-1}}(u,V_{j-1}\setminus W_i)$ corresponds to the cut $S=\set{u}$ in the induced graph $G[(V_j\setminus W_i)\cup \set{u}]$, and thus it can be estimated using a power cut sparsifier.
    There are $O(n)$ such induced graphs, and since they are all fixed in advance (before opening the sketch), taking a union bound we find that a single sample from the power cut distribution is sufficient to approximate cuts within all of them with probability at least $1-n^{-C+1}$.
    The second condition (the out-degree condition) is to shave every super vertex $u\in V_{j-1}$ is modified to be
    \begin{equation}
        \label{eq:outdegree-shaving-condition}
        \rho \cdot \widetilde{\deg}(u) 
        < 
        4\alpha^{-1}w_{\powercut}(E_{G_{j-1}}(u,V_{j-1}\setminus W_i))
        ,
    \end{equation}
    where $\powercut$ is a $(\delta,\epsilon)$-cut sparsifier for $G_{j-1}\{u\cup V_{j-1}\setminus W_i\}$ and $\rho\in(0,1)$ is some constant to be determined later.
    We begin by showing that if $4\alpha^{-1}|E_{G_{j-1}}(u,V_{j-1}\setminus W_i)|>\deg(u)$ then this condition holds.
    By the properties of the $(\delta,\epsilon)$-cut sparsifier
    \begin{align*}
        &4\alpha^{-1}w_{\powercut}(E_{G_{j-1}}(u,V_{j-1}\setminus W_i))
        \\
        &\ge
        4\alpha^{-1} \left( (1-\delta) \cdot |E_{G_{j-1}}(u,V_{j-1}\setminus W_i)|-\epsilon \cdot \deg(u) \right)
        \\
        &\ge 
        4\alpha^{-1} \left( \left( 1-\frac{1}{5\log n} \right) |E_{G_{j-1}}(u,V_{j-1}\setminus W_i)| 
        -\frac{0.01\alpha\beta}{\log^3 n} \deg(u)
         \right)
         \\
         &>
        \left( 1-\frac{1}{5\log n} -\frac{0.04}{\log^3 n}\right) \deg(u)
        \\
        &\ge \left( 1-\frac{1}{5\log n} -\frac{0.04}{\log^3 n}\right) \left( 1+\frac{1}{100} \right)^{-1}\widetilde{\deg}(u)
        ,
    \end{align*}
    where the last inequality follows from the assumption on $u$.
    Therefore, choosing $\rho=3/4$ we have that \Cref{eq:outdegree-shaving-condition} holds for any large enough $n$.
    We now show that if $8\alpha^{-1}|E_{G_{j-1}}(u,V_{j-1}\setminus W_i)|\le \deg(u)$ then the condition does not hold.
    Notice that in this case
    \begin{align*}
        &4\alpha^{-1}w_{H}(E_{G_{j-1}}(u,V_{j-1}\setminus W_i))
        \\
        &\le
        4\alpha^{-1}\left( 1+\frac{1}{5\log n} \right) |E_{G_{j-1}}(u,V_{j-1}\setminus W_i)| 
        +\frac{0.04}{\log^3 n} \deg(u)
        \\
        &\le 
        4\alpha^{-1}\left( 1+\frac{1}{5\log n} \right) \frac{\alpha}{8}\deg(u)
        +\frac{0.04}{\log^3 n} \deg(u)
        \\
        &\le 
        \left( \frac{1}{2}+\frac{1}{10\log n} + \frac{0.04}{\log^3 n}\right) \deg(u)
        \\
        &\le \left( \frac{1}{2}+\frac{1}{10\log n} + \frac{0.04}{\log^3 n}\right)\left( 1-\frac{1}{100} \right)^{-1} \widetilde{\deg}(u)
        \le \frac{3}{4}\widetilde{\deg}(u)
        ,
    \end{align*}
    where the second inequality is by the assumption on $u$ and the last inequality is true for large enough $n$.
    Therefore, we can implement line~\ref{lst:line:contraction-criteria} with slack $1/2$.
\end{proof}

\subsection{Power Cut Sparsifiers with Contraction}
\label{sec:power-cut-sparsifier-contraction}
In this section we show how to construct a power cut sparsifier for a contraction of a graph $G$ given in a dynamic stream.
The construction is a modification of the algorithm of \cite{FKM23}.
Note that while the result of \cite{FKM23} is stated for simple graphs, it is explicitly mentioned that this construction works for unweighted multigraphs as well.
Furthermore, an earlier version of \cite{FKM23} contained a similar result to \Cref{lemma:power-cut-sparsifier-contraction}, which was removed in the most up-to-date version.
The construction was part of a proof for constructing a spanner in the streaming setting that contained an error, and thus was removed.
However, the source of the error was elsewhere, and the expander decomposition for a contraction is still valid.

We begin by presenting the algorithm of \cite{FKM23} for constructing a $(\delta,\epsilon)$-power cut sparsifier in a dynamic stream, and then explain how to modify it to handle contractions.
To construct the sparsifier, sample each edge $(u,v)\in E$ independently with probability $p=\gamma(1/d(u)+1/d(v))$, where $\gamma=O(\frac{\log^2 n}{\epsilon\delta})$.
To implement this algorithm in a dynamic stream, initialize $\log n$ uniform hash functions $h_i:\binom{V}{2}\to \set{0,1}$, which take an edge and return whether it is sampled in the $i$-th level.%
\footnote{Note that naively storing each such function requires $O(n^2)$ space, but one can use a pseudorandom generator to reduce the space to $\tO(n)$.}
These hash functions induce $\log n$ graphs, $G_i=(V,E_i)$ where an edge $e$ is included $E_i$ if $h_j(e)=1$ for all $j\le i$.
Since it is impossible to store all the edges of $G_i$, the algorithm utilizes the following sparse-recovery procedure.
\begin{lemma}[Folklore, e.g. Theorem 2.2 of \cite{AGM12}]
    There exists a linear sketch that recovers a $k$-sparse vector $x\in \R^n$ from a dynamic stream with success probability $1-p$ using $O(k\log n\log 1/p)$ space.
\end{lemma}
For each $v\in V$, the algorithm maintains a sketch $s_v^i$ of the vector $x_v^i\in \mathbb{R}^{\binom{V}{2}}$ where $x_v^i(e)=1$ if $e$ is incident to $v$ in $G_i$, and $0$ otherwise.
For every vertex $v\in V$, let $i^*(v)=\lfloor \log_2(d(v))/(2\gamma) \rfloor$.
It is straightforward to see that with high probability each vertex $v\in V$ has degree at most $8\log n$ in $G_{i^*(v)}$.
Therefore opening the sketch $s_v^{i^*(v)}$ of $x_v^{i^*(v)}$ recovers all edges incident to $v$ in $G_{i^*(v)}$.
Furthermore, these edges were sampled with probability at least $\gamma/d(v)$.
The sparsifier is then formed by taking the union of all edges sampled for each $v\in V$, and giving each edge $e=(u,v)$ weight $2^{\min\left\{ i^*(u), i^*(v) \right\}}$.
One of the main theorems of \cite{FKM23} is that this algorithm yields an $(\delta,\epsilon)$-power cut sparsifier with high probability using $\tO(n/(\epsilon\delta))$ space.
We now show how to modify this algorithm to handle contractions.
\begin{proof}[Proof of \Cref{lemma:power-cut-sparsifier-contraction}]
    To construct the power cut sparsifier we now need an additional cut sparsifier $\cutspar$ of $G$ with quality $\epsilon=1/100$.
    This can be constructed using $\tO(n)$ storage in one pass of streaming using \Cref{theorem:cut-sparsifier-streaming}.

    We now describe the modified algorithm.
    Begin by assigning an arbitrary ordering on the vertices of $G$.
    During the stream perform the same algorithm, except when updating the sketch $s_v^i$ of $x_v^i$, for every edge $e=(u,v)$ insertion set $x_v^i(e)=1$ for the higher index vertex and $x_v^i(e)=-1$ for the lower index vertex.
    Then, at the end of the stream, once the algorithm is given the contraction $V'$ of $V$, let $s_u^i=\sum_{v\in u}s_v^i$ for every $u\in V'$ and $i\in[\log n]$.
    To estimate the degree of $u\in V'$, the algorithm uses the cut sparsifier $\cutspar$ to estimate $\mintcut_G(u)$ which is equal to the degree.
    Then, let $i^*(u)=\lfloor \log_2(d(u))/(2\gamma) \rfloor$ and recover the edges in the sketch $s_u^{i^*(u)}$.
    The rest of the algorithm is identical to the original and the correctness analysis follows similarly.
\end{proof}

{\small
  \bibliographystyle{alphaurl}
  \bibliography{bibliography}
} %
\appendix
\section{Gomory-Hu Trees for Symmetric Submodular Functions}
\label{sec:gomory-hu-submodular}
In this section we prove \Cref{theorem:submodular-gomory-hu}.
We begin by restating the theorem for ease of reference.
\ghsubmodular*
The proof is based on the reduction of \cite{GKYY25} from constructing a Gomory-Hu tree of a graph to $\polylog(n)$ calls to max-flow computations on graphs of total size $\tO(n)$.

\begin{proof}
    The correctness of the algorithm follows from the fact that $g$ is symmetric submodular, as is explicitly shown in \cite{GKYY25}.
    Therefore, it remains to analyze the query complexity of the algorithm.
    The algorithm is based on partitioning the set of terminals, and then recursing on each part.
    The following claim shows that the total number of elements in all the recursive calls is only $O(n\log^2 n)$.
    We note that the claim is stated for weighted graphs, but the same proof applies to symmetric submodular functions as well.
    \begin{claim}[Claim 5.3 of \cite{GKYY25}]
        \label{claim:total-ground-set-size-gomory-hu}
        A call on a function $g$ of ground set size $n$ to the Gomory-Hu tree algorithm, creates recursive calls on functions of total ground set size $O(n\log^2 n)$.
    \end{claim}
    In each recursive call, the algorithm makes either $O(\log^2 n)$ calls to an SFM procedure and $O(\polylog n)$ calls to an isolating-cuts procedure.
    Outside these procedures the algorithm does not make any additional queries to $g$.
    The following theorem implements the isolating-cuts procedure using $O(\log n)$ calls to an SFM procedure.
    \begin{theorem}[Theorem 2.4 of \cite{MN21}]
        Let $g$ be a symmetric submodular function over a ground set $V$ of size $n$, let $R\subseteq V$ be a set of terminals, and $\A$ be an SFM procedure.
        Then, there exists an isolating-cuts procedure for $g$ and $R$ that makes $O(\log |R|)$ calls to $\A$ with ground set size $O(n)$ and an additional $|R|$ calls to $\A$ on ground sets $U_1,\ldots, U_{|R|}$ where $\sum_{i=1}^{|R|} |U_i| \le n$.
    \end{theorem}
    We now show that the function $Q(n)$ is superadditive, i.e., for every $n_1,\ldots,n_k$ such that $\sum_i n_i \le n$ we have $\sum_i Q(n_i) \le Q(n)$.
    To do so, we will use the following fact, whose proof is based on \cite{mathoverflow490528} and is included at the end of the section for completeness.
    \begin{fact}
        \label{fact:superadditive-function}
        Let $f:\R\to\R_+$ be a function such that $f(x)/x$ is non-decreasing, then $f$ is superadditive.
    \end{fact}
    Notice that this holds for $Q(n)$ from the assumption that $Q(n)=n\cdot g(n)$ where $g(n)$ is non-decreasing.
    Therefore, denoting the number of vertices in the $j$-th partition of the $i$-th recursive call by $n_{i,j}$, the total query complexity of the algorithm is
    \begin{equation*}
        \sum_i \sum_j Q(n_{i,j})
        \le  \sum_i Q(n)\cdot \polylog n
        \le Q(n)\cdot \polylog n
        , 
    \end{equation*}
    where the first inequality is since $n_{i,j} \le n$ and $\sum_j n_{i,j} \le n\cdot\polylog n$,
    and the last inequality is by \Cref{claim:total-ground-set-size-gomory-hu}.
\end{proof}
\begin{proof}[Proof of \Cref{fact:superadditive-function}]
    The proof is based on \cite{mathoverflow490528} and is included for completeness.
    Let $f:\R\to\R_+$ be a function such that $f(x)/x$ is non-decreasing.
    Notice that $f(x)\le xf(x+y)/ (x+y)$ for every $x,y>0$ and similarly $f(y)\le yf(x+y)/(x+y)$.
    Therefore,
    \begin{equation*}
        f(x)+f(y)
        \le \frac{(x+y)f(x+y)}{x+y}
        = f(x+y)
        \qedhere
    \end{equation*}
\end{proof}

\section{Max-Flow Cut-Query Lower Bound}
\label{sec:max-flow-cut-query-lower-bound}
In this section we prove \Cref{lemma:maximum-s-t-flow-lower-bound}, proving a lower bound on the number of cut queries needed to recover every edge in some maximum $s,t$-flow in an unweighted graph $G$.
\begin{proof}
The lower bound is based on a graph family closely related to a hard example provided in \cite{KL98}.
We define a family $\G$ of unweighted graphs on $n+2$ vertices with two designated vertices $s,t$.
The family is defined as follows.
Initialize an empty graph $G$ on $n+2$ vertices.
Let $n'\eqdef 8\cdot \lfloor n/8\rfloor$ and ignore the excess $n-n'$ vertices for the rest of the proof.
In addition, let $\f$ be the closest integer to $n'$ such that $n'/\sqrt{\f}$ is integral and $\f$ is divisible by $8$;
note that $\f = \Theta(n)$.
The graph is $G$ is obtained by picking $n'/4$ arbitrary vertices $s_1,\ldots,s_{n'/4}$, and adding the edges $\{s,s_i\}_{i=1}^{n'/4}$ and similarly picking $n'/4$ different vertices $t_1,\ldots,t_{n'/4}$ and adding the edges $\{t,t_i\}_{i=1}^{n'/4}$.
Then, partition the remaining $n'/2$ vertices into layers $L_1,\ldots,L_{n'/\sqrt{\f}}$, where each layer $L_i$ contains $\sqrt{\f}/2$ vertices.
Connect each layer $L_i$ to layer $L_{i+1}$ by sampling uniformly a random set of cardinality $\f/8$ from $L_i\times L_{i+1}$.
Finally, add a complete bipartite graph between the sets $\{s_1,\ldots,s_{n'/4}\}$ and $L_1$, and $\{t_1,\ldots,t_{n'/4}\}$ and $L_{n'/\sqrt{\f}}$.

It is straightforward to see that the maximum $s,t$-flow in $G$ is $\f/8$, since each layer $L_i$ is connected to the next layer $L_{i+1}$ by $\f/8$ edges.
Furthermore, every edge between each consecutive layers $L_i,L_{i+1}$ participates in the maximum flow.
Therefore, to recover every edge in the maximum flow, we must recover all the edges between each pair of layers $L_i,L_{i+1}$.
Notice that the number of bits needed to encode the edges between two consecutive layers is $\log_2 \binom{\f/4}{\f/8} \le \f$.
Since there are $n/\sqrt{\f}$ such layers, the total number of bits needed to encode all the edges in the maximum flow is $n'\sqrt{\f}\ge n^{3/2}/64$.

We now use the above construction to prove a lower bound on the number of cut queries needed to recover all the edges in the maximum flow.
Denote the success probability of the algorithm by $p=\Pr_{G\sim \G}[A=G]$ and recall that $p$ is a constant.
It is clear that recovering every edge of the maximum flow is equivalent to recovering the graph $G$.
To prove the lower bound on the number of cut queries needed to recover $G$, we will use an information-theoretic argument.
By Yao's minimax principle, it suffices to prove a lower bound for deterministic algorithms on a random graph $G$ sampled from the above distribution.
Let $A\in \G$ be the output of the deterministic cut-query algorithm after making a sequence $R$ of $q$ cut queries on a graph $G$ uniformly sampled from $\G$.
It is clear that $H(G')\le H(R) \le 2q\log n$ since each cut query returns a value in the range $[0,n^2]$.
Therefore,
\begin{equation*}
    H(G \mid A) 
    = H(G,A) - H(A)
    \ge H(G) - H(A)
    \ge n^{3/2}/64-2q\log n
    .
\end{equation*}
On the other hand, by Fano's inequality we can upper bound $H(G \mid A)$ as follows,
\begin{equation*}
    H(G \mid A) 
    \le H(\text{error}) + \Pr[\text{error}]\cdot \log_2 |\G|
    \le 1 + (1-p)\cdot n^{3/2}/64
    .
\end{equation*}
Combining the two we bounds we find that $q\ge \Omega(pn^{3/2}/\log n)$.
\end{proof}

\fi

\end{document}